\documentclass{article}

\usepackage[noadjust]{cite}
\usepackage{amsmath,amssymb,amsfonts,amsthm,bbm}
\usepackage{algorithmic}
\usepackage{textcomp}
\usepackage{xcolor}
\def\BibTeX{{\rm B\kern-.05em{\sc i\kern-.025em b}\kern-.08em
    T\kern-.1667em\lower.7ex\hbox{E}\kern-.125emX}}

\usepackage[ansinew]{inputenc}          
\usepackage{psfrag}                     
\usepackage[dvips]{graphicx,gincltex,epsfig}            
\usepackage[all]{xy}
\usepackage{algorithm,algorithmic}
\usepackage{a4wide}
\usepackage[dvips]{epsfig}
\usepackage{enumerate}
\usepackage{xcolor}
\usepackage{tikz}
\usepackage{nccmath}
\usepackage{hhline} 
\makeatletter
\newcommand\HUGE{\@setfontsize\Huge{34}{60}}
\makeatother   
\usepackage{enumitem, hyperref}
\makeatletter
\def\namedlabel#1#2{\begingroup
    #2%
    \def\@currentlabel{#2}%
    \phantomsection\label{#1}\endgroup
}
\makeatother

\theoremstyle{plain}
\newtheorem{theorem}{Theorem}

\newtheorem{lemma}[theorem]{Lemma}
\newtheorem{corollary}[theorem]{Corollary}

\theoremstyle{definition}

\newcommand{\N}{\mathbb{N}}
\newcommand{\F}{\mathbb{F}}

\newcommand{\nolla}{\mathbf{0}}
\newcommand{\bu}{\mathbf{u}}

\newcommand{\bw}{\mathbf{w}}
\newcommand{\s}{\mathbf{s}}
\newcommand{\bc}{\mathbf{c}}

\newcommand{\supp}{\textrm{supp}}

\newcommand{\bx}{\mathbf{x}}
\newcommand{\e}{\mathbf{e}}

\newcommand{\bz}{\mathbf{z}}
\newcommand{\cc}{\mathbf{c}}

\newcommand{\y}{\mathbf{y}}
\newcommand{\by}{\mathbf{y}}

\newcommand{\LL}{\mathcal{L}}
\newcommand{\be}{\mathbf{e}}
\newcommand{\dmin}{d_{\textrm{min}}}

\title{The Levenshtein's Sequence Reconstruction Problem and the Length of the List\thanks{This paper was presented in part at the $2022$ IEEE International Symposium on Information Theory (ISIT2022) (see \cite{junnila2022list}).}~
	\thanks{The authors were funded in part by the Academy of Finland grant 338797 and the third author was funded by the Finnish Cultural Foundation.}
}

\author{Ville Junnila, Tero Laihonen and Tuomo Lehtil{\"a}%
\thanks{The authors are currently with the Department of Mathematics and Statistics,
University of Turku, Turku FI-20014, Finland (e-mail: viljun@utu.fi;
terolai@utu.fi; tualeh@utu.fi).}~\thanks{Much of the work of the third author was done when he was in Univ Lyon, Universit\'e Claude Bernard, CNRS, LIRIS - UMR 5205, F69622, France.}%
}

\begin{document}

\maketitle

\begin{abstract}
In the paper, the Levenshtein's sequence reconstruction problem is considered in the case where at most $t$ substitution errors occur in each of the $N$ channels and the decoder outputs a list of length $\LL$. Moreover, it is assumed that the transmitted words are chosen from an $e$-error-correcting code $C \ (\subseteq \{0,1\}^n)$. Previously, when $t = e+\ell$ and the length $n$ of the transmitted word is large enough, the  numbers of required channels are determined for $\LL =1, 2 \text{ and } \ell+1$. Here we determine the exact number of channels in the cases $\LL = 3, 4, \ldots, \ell$. Furthermore, with the aid of covering codes, we also consider the list sizes in the cases where the length $n$ is rather small (improving previously known results). After that we study how much we can decrease the number of required channels when we use list-decoding codes. Finally, the majority algorithm is discussed for decoding in a probabilistic set-up; in particular, we show that with high probability a decoder based on it is verifiably successful, i.e., the output word of the decoder can be verified to be the transmitted one. 
\end{abstract}

\noindent\textbf{Keywords:}
	Information Retrieval, Levenshtein's Sequence Recontruction,  List Decoding, Majority Algorithm, Sauer-Shelah Lemma, Substitution Errors.


\section{Introduction}

In this paper, the Levenshtein's \emph{sequence reconstruction problem}, which was  introduced in \cite{Levenshtein}, is studied when the errors are substitution errors.  For many related sequence reconstruction problems (concerning, for instance, deletion and insertion errors) consult, for example, \cite{Levenshtein,levenshtein2005reconstruction,gabrys2018sequence, horovitz2018reconstruction,  Maria_Abu-Sini, Uusi_Maria_Abu-Sini, levenshtein2001efficient}. Originally, the motivation for the sequence reconstruction problem came from biology and chemistry where the familiar redundancy method of error correction is not suitable. The sequence reconstruction problem has come back to the focus, since  it was recently pointed out that this problem is highly  relevant to information retrieval in advanced storage technologies where the stored information is either a single copy, which is read many times, or the stored information has several copies \cite{horovitz2018reconstruction, yaakobi2016constructions}. This problem (see \cite{horovitz2018reconstruction}) is especially applicable to DNA data storage systems (see \cite{bornholt2016dna, church2012next, grass2015robust, yazdi2015dna}) where DNA strands provide numerous erroneous copies of the information and the goal is to recover the information using these copies.

Let us denote the set $\{1,2,\dots, n\}$ by $[1,n].$ Denote by $\F$ the finite field of two elements, and denote the binary Hamming space by $\F^n$.  The  \emph{support} of the word $\bx=x_1\dots x_n\in \F^n$ is defined via  $\supp(\bx)=\{i\mid x_i\neq 0\}.$ Let us denote the all-zero word $\nolla=00\dots 0\in \F^n$ and by $\be_i\in \F^n$ a word with $1$ in the $i$th coordinate and zeros elsewhere.  The \emph{Hamming weight} $w(\bx)$ of $\bx\in \F^n$ is $|\supp(\bx)|$. The \emph{Hamming distance} is defined as
$d(\bx,\by)=w(\bx+\by)$ for $\bx,\by\in\F^n$. Let us denote the \emph{Hamming ball} of radius $t$  centered
at $\bx\in \F^n$ by $B_t(\bx)=\{\by\in \F^n\mid d(\bx,\by)\le t\}$ and its cardinality by $V(n,t)=\sum_{i=0}^t\binom{n}{i}$.  A nonempty subset of $\F^n$ is called a \emph{code} and its elements are called \emph{codewords}. The \emph{minimum distance} of a code
$C\subseteq \F^n$ is defined as $ \dmin(C)=\min_{\bc_1,\bc_2\in C, \bc_1\neq \bc_2}d(\bc_1,\bc_2).$ Consequently, the code $C$ has the error-correcting capability $e=e(C)=\lfloor(\dmin(C)-1)/2 \rfloor.$

Next we consider the sequence reconstruction problem. For the rest of the paper, let $C\subseteq \F^n$ be any $e$-error-correcting code.  A codeword $\bx\in C$ is transmitted through $N$ channels where, in each of them, at most $t$ substitution errors can occur. In the sequence reconstruction problem, our aim is to reconstruct $\bx$ based on   the $N$ distinct outputs $Y=\{\by_1,\dots, \by_N\}$ from the
channels (see Fig.~\ref{LevenshteinFig}).   

\begin{figure}[htp]
	\centering
	\includegraphics[scale=0.33,trim=1500 270 1500 60,clip]{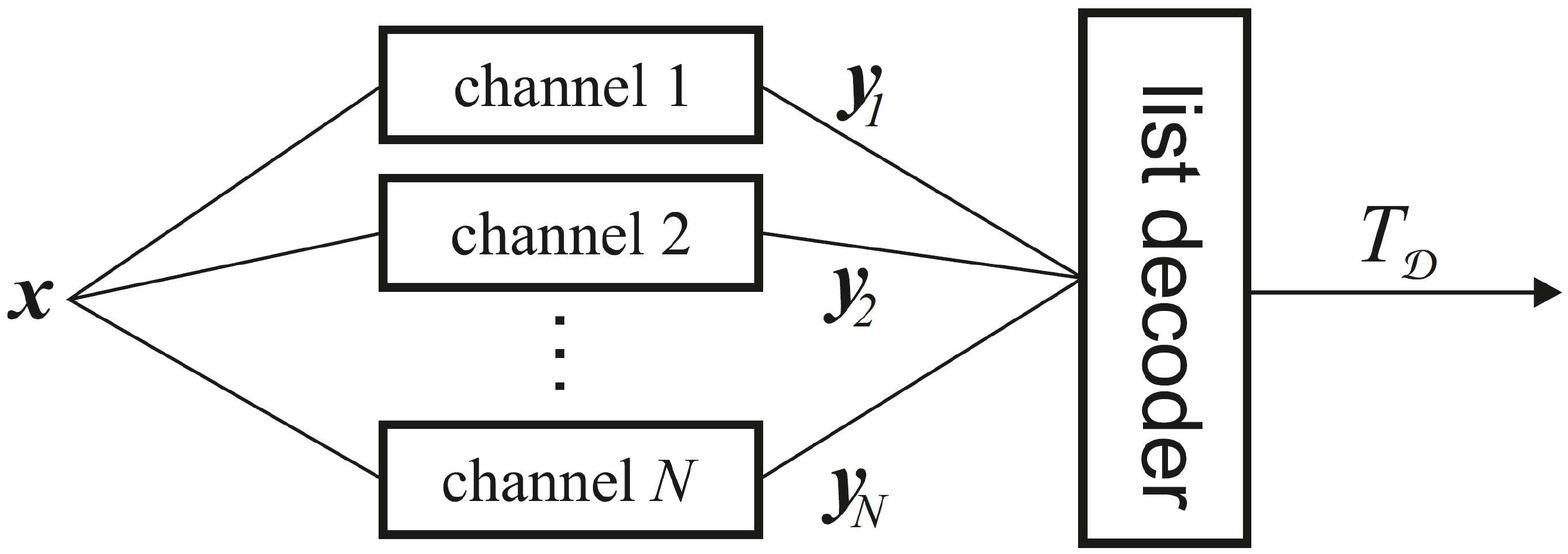}
	\centering\caption{The Levenshtein's sequence reconstruction.} \label{LevenshteinFig}
\end{figure}

It is assumed that $t>e(C)$ (if $t\le e(C),$ then only one channel is enough to reconstruct $\bx$). For $\ell\ge 1$, let us denote  $$t=e(C)+\ell=e+\ell$$ for the rest of paper. The situation where we obtain sometimes a short list of possibilities for $\bx$ instead of always recovering $\bx$ uniquely,  is considered in \cite{yaakobi2018uncertainty,YBiram}. Based on the set $Y$ and the code $C$, the list decoder (see  Fig.~\ref{LevenshteinFig}) $\mathcal{D}$ gives an estimation $T_\mathcal{D}=T_\mathcal{D}(Y)=\{\bx_1,\dots,\bx_{|T_\mathcal{D}|}\}$  on the sequence $\bx$ which we try to reconstruct. We denote by $\mathcal{L_D}$ the maximum cardinality of the list $T_\mathcal{D}(Y)$ over all possible sets $Y$ of output words. The decoder is said to be  \emph{successful} if  $\bx\in T_\mathcal{D}$. In this paper, we focus on the smallest possible value of $\mathcal{L_D}$ over all successful decoders $\mathcal{D}$, in other words, on  $\mathcal{L}=\min_{\mathcal{D}\text{ is successful}}\{\mathcal{L_D}\}$. Let us denote $$T=T(Y)=C\cap\left(\bigcap_{\y\in Y}B_t(\y)\right).$$ Consequently, $$\mathcal{L}=\max\{|T(Y)|\mid Y \text{ is a set of } N \text{ output words}\}.$$ The  value $\mathcal{L}$ depends on $e,\ell,n$ and $N$. Obviously, one would like to have as small $\mathcal{L}$  as possible. Observe that we consider the worst case scenario of the output channels regarding $\LL$. In such a situation, the channels are sometimes called adversarial; for example, see~\cite{cheraghchi2020coded}. The problem of minimizing $\LL$ is studied, for example, in \cite{yaakobi2018uncertainty,YBiram,JLirvnic,JLcirsu,laihonen2017improved,laihonen2019improved}. 
 Also a probabilistic versions of this problem have been studied (often under the name \textit{trace reconstruction}) for example in \cite{batu2004reconstructing, viswanathan2008improved}. In this paper, we mainly consider the relation between $N$ and $\LL$ for various $n$ after we  fix  two parameters $\ell$ and $e$ (while letting $C$ be any $e$-error-correcting code). The sequence reconstruction problem is also closely related  (see \cite{JLirvnic}) to  \emph{information retrieval in 	associative memory} introduced by Yaakobi and Bruck \cite{yaakobi2018uncertainty, YBiram}.  

\medskip

The structure of the paper is as follows. In Section~\ref{SectionBasics}, we recall some of the known results. In particular, it is pointed out that if we have at least (resp. less than)  $V(n,\ell-1)+1$ channels, then the list size is constant with respect to $n$ (resp. there are $e$-error-correcting codes with list size depending on $n$). In Section \ref{SectionMoreChannels}, we give the complete correspondence between the list size and the number of channels when we have more than $V(n,\ell-1)+1$ channels and $n$ is large enough. It is sometimes enough to increase the number of channels only by a constant amount in order to decrease the list size (see Corollary~\ref{Cor:lLista}). Section \ref{SectionCoveringCodes} focuses on improving the bounds on the list size when $n$ is not restricted and we obtain strictly more channels than $V(n,\ell-1)+1$. Section \ref{SectionLessChannels} is devoted to list size when we have \emph{less} than $ V(n,\ell-1)+1$ channels. The final section deals with the reconstruction with the aid of a majority algorithm on the coordinates among the output words in $Y$. 

\section{Known results} \label{SectionBasics}

In this section we present some known results on how the two values $N$ and $\LL$ are linked. The basic idea on estimating  $\LL$ is the following: we analyse the maximum number of output words ($N$) we can fit in the intersection of $\LL$ $t$-radius balls centered at codewords. As expected, the length $\LL$ of the outputted list strongly depends on the number of channels.

Previously, in~\cite{Levenshtein} and \cite{yaakobi2018uncertainty}, the problem has been considered for $\LL = 1$ and $\LL = 2$, respectively. Moreover, in \cite{junnila2020levenshtein}, the exact number of channels $N$ required to have $\LL$ constant on $n$ has been presented, see Theorems \ref{shatter raja} and \ref{RemarkNonConstantList}. Following theorem gives an exact number of channels required to have $\LL=1$.
\begin{theorem}[\cite{Levenshtein}]\label{L=1}
	We have $\LL\leq1$ if $$N\geq \sum_{i=0}^{\ell-1}\binom{n-2e-1}{i}\sum^{t-i}_{k=e+1+i-\ell}\binom{2e+1}{k}+1.$$
\end{theorem}

\begin{theorem}[\cite{yaakobi2018uncertainty}]\label{YB raja l=2}
If
$N\geq\sum_{i_1,i_2,i_3,i_4}\binom{n-\lceil\frac{3d}{2}\rceil}{i_1}\binom{\lceil\frac{d}{2}\rceil}{i_2}\binom{\lceil\frac{d}{2}\rceil}{i_3}\binom{\lfloor\frac{d}{2}\rfloor}{i_4}+1$ for \begin{itemize}
\item $0\leq i_1\leq t-\lceil\frac{d}{2}\rceil$,
\item $i_1+\lfloor\frac{d}{2}\rfloor-t\leq i_4\leq t-\lceil\frac{d}{2}\rceil-i_1$, 
\item $2\lceil\frac{d}{2}\rceil-t+i_1\leq i_3\leq t-(i_1+i_4)$ and 
\item$\max\{i_1-i_3-i_4+\lceil\frac{3d}{2}\rceil-t,i_1+i_3+i_4+\lceil\frac{d}{2}\rceil-t\}\leq i_2\leq t-(i_1+i_4+\lceil\frac{d}{2}\rceil-i_3)$,
\end{itemize}  then $\mathcal{L}\leq2$ for any code $C$ with minimum distance $d$.
\end{theorem}

The next result is a a reformulation of a result by Yaakobi and Bruck \cite[Algorithm~18]{yaakobi2018uncertainty} proven in \cite{junnila2020levenshtein}.\nopagebreak
\begin{theorem}\label{2 kanavaa 2l yli l}
	Let $n \geq 2\ell-1$ and $C$ be an $e$-error-correcting code in $\F^n$. If $N \geq V(n,\ell-1)+1$, then we have $$\LL \leq \binom{2\ell}{\ell}.$$
\end{theorem}

The bound in Theorem \ref{2 kanavaa 2l yli l} can be improved to $2^\ell$ which has been shown to be tight in \cite{junnila2020levenshtein}.
\begin{theorem}[Theorem $7$, \cite{junnila2020levenshtein}]\label{shatter raja}
	Let $n \geq \ell$ and $C$ be an $e$-error-correcting code in $\F^n$. If $t = e + \ell$ and $N \geq V(n,\ell-1)+1$, then we have $$\LL \leq 2^{\ell}.$$
\end{theorem}

Besides the $2^\ell$ part, also the value $V(n,\ell-1)+1$ for the number of channels is tight, that is, if the value for $N$ is less, then list size $\LL$ can be linear with respect to  $n$.
\begin{theorem}[Theorem $10$, \cite{junnila2020levenshtein}]\label{RemarkNonConstantList}
	If $N\leq V(n,\ell-1)$, then there exists an $e$-error-correcting code such that
	$\LL \geq \lfloor n / (e+1) \rfloor$. 
\end{theorem}

Let us denote for the rest of the paper \begin{align*}
n(e,\ell,b)=(\ell-1)^2\left(b-e+(e+1)\left(b-3e-2e^2+eb+\binom{b-2e-1}{2}\right)\right)+\ell-2.
\end{align*} Although the bound for $\LL$ in Theorem \ref{shatter raja} cannot be improved in general, we can improve it, when $n$ is large, to $\ell+1$.
\begin{theorem}[Theorem $20$, \cite{junnila2020levenshtein}]\label{l+1Theorem POL}
	Let $n\geq n(e,\ell,b)$, $b=\max\{3t,4e+4\}$, $|Y|=N\geq V(n,\ell-1)+1$ and $C$ be an $e$-error-correcting code. Then we have $$\mathcal{L}\leq \ell+1.$$
\end{theorem}

Moreover, the bound $\ell+1$ is tight.
\begin{theorem}[Theorem $9$, \cite{junnila2020levenshtein}]\label{l+1 tiukka}
 There exists an $e$-error-correcting code $C\subseteq \F_2^n$ such that $\mathcal{L}\ge \ell+1$ if $n\ge \ell+\ell e+e$ and the number of channels satisfies $N\le V(n,\ell-1)+1$.
\end{theorem}

Finally, in \cite[Theorem $6$]{yaakobi2018uncertainty}, the authors have given exact number of channels required to have $\LL\leq 2$. All in all,  $N$ is known precisely only for these three values when $\LL$ is constant on $n$. In the following section, we give the missing values for $N$.

\section{List size with more channels}\label{SectionMoreChannels}

In this section, we give exact bounds for the number of channels $N=N_h+1$ (when $n$ is large) which is required for satisfying $\LL<h$ for every constant value $h$. Previously, $N_h$ was known only for three values $h=2,3,\ell+2$. To achieve this, we need to introduce two technical lemmas from \cite{junnila2020levenshtein}.

In the following lemma, when $n$ is large, it is shown that if any three codewords in $T(Y)$ differ within some subset of coordinates $\overline{D}$ of constant size $b$, then there exists an output word $\by$ which differs from these codewords in at least $\ell-1$ coordinate positions outside of $\overline{D}$. Notice that $\supp(\bw+\bz)$ gives the set of coordinates in which $\bw$ and $\bz$ differ.
\begin{lemma}[Lemma $18$, \cite{junnila2020levenshtein}]\label{l-1 ulkona}
Let $b\geq3t$ be an integer with $t=e+\ell$ and $C_1$ be an $e$-error-correcting code. Assume that $n\geq n(e,\ell,b)$, $|Y|=N\geq V(n,\ell-1)+1$, $|T(Y)|\geq3$ and $\bc_0,\bc_1,\bc_2\in  T(Y)$. If now $\overline{D}\subseteq [1,n]$ is a set such that $|\overline{D}|= b$ and $$\supp(\bc_0+\bc_1) \cup \supp(\bc_0+\bc_2) \cup \supp(\bc_1+\bc_2) \subseteq \overline{D} \text,$$ then for any word $\bw \in \F^n$ we have $\supp(\bw+\cc_0)\setminus \overline{D} = \supp(\bw+\cc_1)\setminus \overline{D} = \supp(\bw+\cc_2)\setminus \overline{D}$ and there exists an output word $\y\in Y$ such that $$|\supp(\y+\cc_0)\setminus \overline{D}|\geq\ell-1.$$
\end{lemma}

The following lemma shows that the distance between any codewords in $T(Y)$ is either $2e+1$ or $2e+2$.
\begin{lemma}[Lemma $19$, \cite{junnila2020levenshtein}]\label{Listasanat lahella POL}
Let $n\geq n(e,\ell,3t)$, $|Y|=N\geq V(n,\ell-1)+1$, $C$ be an $e$-error-correcting code and $|T(Y)|\geq3$. Then we have $d(\bc_1,\bc_2)\leq 2e+2$ for any two $\bc_1,\bc_2\in T(Y)$.
\end{lemma}

We denote by $N'(n,\ell,e,h)$ 
the maximum number of $t$-error channels such that there exists a set of output words $Y\subseteq \F^n$ satisfying $|Y|=N'(n,\ell,e,h)$ and $|T(Y)|\geq h$ for some $e$-error correcting code $C$. By Theorems \ref{RemarkNonConstantList} and \ref{l+1Theorem POL}, $N'(n,\ell,e,h)=V(n,\ell-1)$ for all $\ell+2\leq h\leq \lfloor n/(e+1)\rfloor$ when $n$ is large enough. Hence, when we use notation $N(n,\ell,e,h)$, we assume that $h$ is the smallest integer for which $N'(n,\ell,e,h)=N'(n,\ell,e,h')$. When the exact formulation is not necessary for clarity, we denote $N(n,\ell,e,h)=N_h$. Observe that, if $N\geq N_h+1$, then $\LL< h$ for all $e$-error-correcting codes. In particular, the difference between $N'(n,\ell,e,h)$ and $N(n,\ell,e,h)$ is that $N'(n,\ell,e,h)$ exists for each value of $h$ but may give the same values for different choices of $h$ while $N(n,\ell,e,h)$ does not exist for every choice of $h$ but each value it attains is unique.

We need the following two technical notations in the next theorem. Let us have $N_h$ channels, then
\begin{align*}
W_w=&\left\{(i_1,\dots,i_h)\mid \text{for each } j:\right.\\ 
&\left. 
i_j\in \N, e+1\geq i_j\geq \frac{w+1-\ell}{2}\text{ and } w\geq \sum_{j=1}^h i_j\right\}
\end{align*}
and
\begin{align*}
W'_w=&\left\{(i_1,\dots,i_h)\mid \text{each }i_j\in \N, e\geq i_1\geq \frac{w-\ell}{2} \text{ and for }\right.\\  
&\left.
j\geq2: e+1\geq i_j\geq \frac{w+1-\ell}{2}\text{ and } w\geq \sum_{j=1}^h i_j\right\}.
\end{align*}

In the following theorem, we give the maximum number of channels $N_h$ which gives list size $\LL= h$ for some $e$-error-correcting code and if $N>N_h$, then $\LL<h$ for all $e$-error-correcting codes.

\begin{theorem}\label{KanavalukemaUusi}
Let $n\geq  n(e,\ell,b)$, $b\geq\max\{3t,4e+4\}$, $\ell\geq2$, $3\leq h\leq\ell+1$. Then \begin{align*}
&N(n,\ell,e,h)= N_h 
= V(n,\ell-1)+\\
&\max\left\{\sum_{w\geq\ell}\sum_{(i_1,\dots,i_h)\in W_w}\binom{n-h(e+1)}{w-\sum_{j=1}^h i_j}\prod_{j=1}^h\binom{e+1}{i_j},\right.\\
%
&\left.
\sum_{w\geq\ell}\sum_{(i_1,\dots,i_h)\in W'_w}\binom{n+1-h(e+1)}{w-\sum_{j=1}^h i_j}\binom{e}{i_1}\prod_{j=2}^h\binom{e+1}{i_j}\right\}.
\end{align*}\end{theorem}
\begin{proof}
Let us have $N=N(n,\ell,e,h)$, $n\geq n(e,\ell,b)$, $b\geq\max\{3t,4e+4\}$, $\ell\geq2$ and $3\leq h\leq\ell+1$. Moreover, let $C$ be such an $e$-error-correcting code, that it maximizes $\LL$ (we have $\LL= h$) when $N=N(n,\ell,e,h)$. Furthermore, let $Y$ be a set of outputs such that $|T(Y)|=\LL=h$ and let us denote $T(Y)=\{\bc_1,\dots,\bc_h\}$.

By Theorem \ref{RemarkNonConstantList}, we have $N\geq V(n,\ell-1)+1$.  Since $C$ is an $e$-error-correcting code and by Lemma \ref{Listasanat lahella POL}, we have $d(\bc_i,\bc_j)\in\{2e+1,2e+2\}$ for each $i\neq j$. Since we are considering binary Hamming space and $h\geq 3$, all pairwise distances cannot be $2e+1$. Let us assume without loss of generality that $d(\bc_1,\bc_2)=2e+2$ and let us then translate the Hamming space so that $\bc_1=\nolla$. Now $w(\bc_2)=2e+2$ and $w(\bc_3)\in\{2e+1,2e+2\}$. Moreover, $|\supp(\bc_2)\cap\supp(\bc_3)|=e+1$. Let $\overline{D}$ be any subset of $[1,n]$ satisfying $\supp(\bc_1+\bc_2) \cup \supp(\bc_1+\bc_3) \cup \supp(\bc_2+\bc_3) \subseteq \overline{D}$ and $|\overline{D}|=b$. Observe that $\supp(\bc_1)\setminus \overline{D}=\supp(\bc_2)\setminus \overline{D}=\supp(\bc_3)\setminus \overline{D}=\emptyset$.

By Lemma \ref{l-1 ulkona}, there exists an output word $\by\in Y$ such that $|\supp(\bc_1+\by)\setminus \overline{D}|\geq\ell-1$. Since $d(\by,\bc_2)\leq t$, we have $|\supp(\by)\cap \supp(\bc_2)|\geq e+1$. Moreover, since $d(\by,\bc_1)\leq t$, we have $w(\by)\leq t$ and hence, $|\supp(\by)\cap \supp(\bc_2)|= e+1$.  Thus, $\supp(\by)=(\supp(\by)\cap \supp(\bc_2))\cup(\supp(\by)\setminus \overline{D})$ and $\supp(\by)\cap(\supp(\bc_3)\setminus \supp(\bc_2))=\emptyset$. Hence, $\supp(\by)\cap \supp(\bc_3)\subseteq \supp(\bc_2)$ and moreover, $\supp(\by)\cap(\supp(\bc_2)\setminus \supp(\bc_3)) = \emptyset$ as otherwise $d(\by,\bc_3) \geq (\ell-1+1)+(2e+1-e) = t+1 > t$ (a contradiction). Together these give that $\supp(\by)\cap \overline{D}=\supp(\bc_2)\cap\supp(\bc_3)$. Notice that for each $i\in[4,h]$ we may choose $\overline{D}$ in such a way that also $\supp(\bc_i)\subseteq \overline{D}$ since $|\overline{D}| = b \geq 4e+4$. Thus, there exists an output word $\by'\in Y$ such that $|\supp(\bc_1+\by')\setminus \overline{D}|\geq\ell-1$. Therefore, as above, $\supp(\by')\cap\overline{D}=\supp(\bc_2)\cap\supp(\bc_i)$ and $\supp(\by')\cap \overline{D}=\supp(\bc_2)\cap\supp(\bc_3)$ implying $\supp(\bc_2)\cap \supp(\bc_i)=\supp(\bc_2)\cap\supp(\bc_3)$. Finally, translate  the Hamming space so that the word $\bz$ with $\supp(\bz)=\supp(\bc_2)\cap \supp(\bc_3)$ becomes $\bz=\nolla$. Then we have $w(\bc_i)\in\{e,e+1\}$ (for $i\in[1,h]$) and $\supp(\bc_i)\cap\supp(\bc_j)=\emptyset$ for each $i \neq j$ since $d(\bc_i,\bc_j)\in\{2e+1,2e+2\}$. Moreover, at most one of $\bc_i$ can have weight $e$ by the minimum distance of $C$.

Let us then count the number of words in $\bigcap_{i=1}^{h}B_t(\bc_i)$. Observe that if $|Y|>|\bigcap_{i=1}^{h}B_t(\bc_i)|$, then we cannot have $h$ codewords in $T(Y)$ but if $|Y|=|\bigcap_{i=1}^{h}B_t(\bc_i)|$, then we can have $h$ codewords in $T(Y)$. Clearly, each word $\by$ with $w(\by)\leq \ell-1$ belongs to the intersection contributing $V(n,\ell-1)$ words to it. Assume then that $w(\by)=w\geq \ell$. As $d(\by,\bc_j)\leq t$ for all $j \in [1, h]$, we have $w(\by)+w(\bc_j) - 2|\supp(\by) \cap \supp(\bc_j)| \leq t$. Denote $i_j=|\supp(\by)\cap\supp(\bc_j)|$. Assume first that $w(\bc_j)=e+1$ for all $j$. Then $\by\in B_t(\bc_j)$ if and only if we have $w+e+1-2i_j\leq t$. Hence, $e+1\geq i_j\geq (w+1-\ell)/2$. Moreover, $\sum_{j=1}^h i_j\leq w$ since $w(\by)=w$ and $\supp(\bc_{j_1})\cap\supp(\bc_{j_2})=\emptyset$ for each $j_1\neq j_2$. In other words, $\by\in \bigcap_{i=1}^{h}B_t(\bc_i)$ if and only if $(i_1,\dots, i_h)\in W_w$. Hence, there exist $\sum_{(i_1,\dots,i_h)\in W_w}\binom{n-h(e+1)}{w-\sum_{j=1}i_j}\prod_{i_j=1}^{h}\binom{e+1}{i_j}$ words of weight $w\geq\ell$ in $\bigcap_{i=1}^h B_t(\bc_i)$.

In the case where $w(\bc_k)=e$ for some $k$, say $k=1$, we have $e\geq i_1\geq (w-\ell)/2$. Thus, $\by\in \bigcap_{i=1}^{h}B_t(\bc_i)$ if and only if $(i_1,\dots, i_h)\in W_w'$. Hence, there exist $\sum_{(i_1,\dots,i_h)\in W_w'}\binom{n+1-h(e+1)}{w-\sum_{j=1}i_j}\binom{e}{i_1}\cdot$ $\prod_{i_j=2}^{h}\binom{e+1}{i_j}$ words of weight $w\geq\ell$ in $\bigcap_{i=1}^h B_t(\bc_i)$. Together these give the claim.\end{proof}

Observe by the proof that the bounds given in Theorem \ref{KanavalukemaUusi} are tight. Notice that geometrically the output sets giving maximal list size are more complicated than, for example, in Theorem \ref{RemarkNonConstantList} (where a ball of volume $V(n,\ell-1)$ is essential). If we increase $N$ by one, then $\LL$ decreases by one since we cannot place the output word within the intersection of $t$-balls centered at codewords in $T(Y)$. Another  observation is that although the sums do not include an upper bound for $w$, there is one. Namely the definition, for $W_w$, gives that $w \leq 2e+\ell+1$ and for $W_w'$ that $w \leq 2e+\ell$.

In the following theorem, we improve the previous result by showing that the two binomial sums within the $\max$ are actually equal.

\begin{theorem}\label{KanavalukemaEkviv}
Let $n\geq  n(e,\ell,b)$, $b\geq\max\{3t,4e+4\}$, $\ell\geq2$, $3\leq h\leq\ell+1$. Then \begin{align*}
&N(n,\ell,e,h) - V(n,\ell-1)\\
&=\sum_{w\geq\ell}\sum_{(i_1,\dots,i_h)\in W_w}\binom{n-h(e+1)}{w-\sum_{j=1}^h i_j}\prod_{j=1}^h\binom{e+1}{i_j}\\
&=\sum_{w\geq\ell}\sum_{(i_1,\dots,i_h)\in W'_w}\binom{n+1-h(e+1)}{w-\sum_{j=1}^h i_j}\binom{e}{i_1}\prod_{j=2}^h\binom{e+1}{i_j}.
\end{align*}\end{theorem}
\begin{proof}
Observe that the claim follows from Theorem \ref{KanavalukemaUusi} if we can prove that the two binomial sums in the claim are equal. For that, we will be considering subsums \begin{align}\label{eqSa}
S_a=&\sum_{w=\ell+2a}^{\ell+2a+1}\sum_{(i_1,\dots,i_h)\in W_w}\left(\binom{n-h(e+1)}{w-\sum_{j=1}^h i_j}
%
\prod_{j=1}^h\binom{e+1}{i_j}\right)
\end{align}
and 
\begin{align}\label{eqSa'}S'_a=&\sum_{w=\ell+2a}^{\ell+2a+1}\sum_{(i_1,\dots,i_h)\in W'_w}\left(\binom{n+1-h(e+1)}{w-\sum_{j=1}^h i_j}
%
\binom{e}{i_1}\prod_{j=2}^h\binom{e+1}{i_j}\right).
\end{align}

We claim that $S_a=S'_a$ for each non-negative integer $a$. When $w=\ell+2a$ we have \begin{align*}
&W_{\ell+2a}=\left\{(i_1,\dots,i_h)\mid \text{for each } j:
%
i_j\in \N, e+1\geq i_j\geq a+1\text{ and } \ell+2a\geq \sum_{j=1}^h i_j\right\}
\end{align*}
and
\begin{align*}
&W'_{\ell+2a}=\left\{(i_1,\dots,i_h)\mid \text{each }i_j\in \N, e\geq i_1\geq a \text{ and for }\right.\\  
&\left.
j\geq2: e+1\geq i_j\geq a+1\text{ and } \ell+2a\geq \sum_{j=1}^h i_j\right\}.
\end{align*}
Moreover, for $w=\ell+2a+1$ we have \begin{align*}
&W_{\ell+2a+1}=\left\{(i_1,\dots,i_h)\mid \text{for each } j:
%
i_j\in \N, e+1\geq i_j\geq a+1\text{ and } \ell+2a+1\geq \sum_{j=1}^h i_j\right\}
\end{align*}
and
\begin{align*}
&W'_{\ell+2a+1}=\left\{(i_1,\dots,i_h)\mid \text{each }i_j\in \N, e\geq i_1\geq a+1 \text{ and for }\right.\\  
&\left.
j\geq2: e+1\geq i_j\geq a+1\text{ and } \ell+2a+1\geq \sum_{j=1}^h i_j\right\}.
\end{align*}

Assume now that the codewords $\bc_i\in T(Y)$ are arranged as in the proof of Theorem \ref{KanavalukemaUusi}, that is, each codeword has weight of $e$ or $e+1$. Hence, their supports do not intersect. Recall that we have at most one codeword with weight $e$ and if that word exists, then we are using set $W'_w$ in our binomial sum. For further rearranging of the Hamming space, we assume that if each word has weight $e+1$, then $\supp(\bc_i)=[(e+1)(i-1)+1,(e+1)i]$ and if there exists a word with weight $e$, it is denoted by $\bc'_1$ (replacing $\bc_1$)  and we have $\supp(\bc'_1)=[2,e+1]$. 

Recall that $w$ describes the weight of a word $\y_i\in \bigcap_{\bc\in T(Y)} B_t(\bc)$ in the proof of Theorem \ref{KanavalukemaUusi} within the notations $W_w$ and $W'_w$. When $w(\y_i)=w=2a+\ell$ and we are considering case $S_a$, then $|\supp(\y_i)\cap \supp(\bc_j)|=i_j\geq a+1$. Let us denote by $Y_{2a}$, $Y_{2a+1}$, $Y'_{2a}$ and $Y'_{2a+1}$ the sets of output words contributing to the sums $S_a$ and $S'_a$, respectively, where $Y_{w}$ (resp. $Y'_{w}$) contains output words of weight $w+\ell$. Let us denote by $D=[h(e+1)+1,n]$.

We will construct the proof by first showing that $Y_{2a}\subseteq Y'_{2a}$, then that $Y'_{2a+1}\subseteq Y_{2a+1}$ and finally that $|Y'_{2a}\setminus Y_{2a}|=|Y_{2a+1}\setminus Y'_{2a+1}|$. Together, these imply that $S_a=S_a'$.

Assume that $\y\in Y_{2a}$. Thus, $w(\y)=\ell+2a$ and $|\supp(\y)\cap \supp(\bc_i)|\geq a+1$ for each $i\in[1,h]$. Hence, $e\geq|\supp(\y)\cap \supp(\bc'_1)|\geq a$ since $\supp(\bc_1)=\supp(\bc'_1)\cup\{1\}$. Hence, $\y\in Y'_{2a}$.

Assume then that $\y\in Y'_{2a+1}$. Thus, $w(\y)=\ell+2a+1$ and $|\supp(\y)\cap \supp(\bc_i)|\geq a+1$ for each $i\in[2,h]$ and $|\supp(\y)\cap \supp(\bc'_1)|\geq a+1$. Hence, $|\supp(\y)\cap \supp(\bc_1)|\geq a+1$ and $\y\in Y_{2a+1}$. Therefore, $Y_{2a}\subseteq Y'_{2a}$ and $Y'_{2a+1}\subseteq Y_{2a+1}$.

Let us now consider output word $\y'\in Y'_{2a}\setminus Y_{2a}$. We again have $w(\y')=\ell+2a$, $|\supp(\y')\cap \supp(\bc_j)|=i_j\geq a+1$ for each $j\in[2,h]$. However, $|\supp(\y')\cap \supp(\bc_1)|=a=|\supp(\y')\cap \supp(\bc'_1)|$.  Observe that now especially $\supp(\y')\cap \{1\}=\emptyset$. Thus, $\y'+\e_1\in Y_{2a+1}$. Let us denote $i_D=|\supp(\y')\cap D|$. We have $i_D=w(\y')-a-\sum_{j=2}^h i_j=\ell+a-\sum_{j=2}^h i_j$. Consider then output word $\y\in Y_{2a+1}\setminus Y'_{2a+1}$. We again have $w(\y)=\ell+2a+1$, $|\supp(\y)\cap \supp(\bc_j)|=i_j\geq a+1$ for each $j\in[1,h]$. However, $|\supp(\y)\cap \supp(\bc'_1)|=a$; in particular, we have $1\in \supp(\y)$. Indeed, if $1\not\in \supp(\y)$, then $\y\not\in Y_{2a+1}$ and if $|\supp(\y)\cap \supp(\bc'_1)|\geq a+1$, then $\y\in Y'_{2a+1}$. Thus, $\y+\e_1\in Y'_{2a}$. Moreover, we have $i_D=|\supp(\y')\cap D|=w(\y)-a-1-\sum_{j=2}^h i_j=\ell+a-\sum_{j=2}^h i_j$. Therefore, for each output word $\y'\in Y'_{2a}\setminus Y_{2a}$, we have $\y'+\e_1\in Y_{2a+1}\setminus Y'_{2a+1}$ and for each output word $\y\in Y_{2a+1}\setminus Y'_{2a+1}$, we have $\y+\e_1\in Y'_{2a}\setminus Y_{2a}$. Thus, $|Y'_{2a}\setminus Y_{2a}|=|Y_{2a+1}\setminus Y'_{2a+1}|$.

Now, we have $S_a=S'_a$ and the claim follows.
\end{proof}

In the following theorem, we show that  $N_h$ exists and is unique for each value $h\in[3,\ell+1]$ when $n$ is large. As we have seen previously, this does not hold when $h\geq\ell+2$.

\begin{theorem}\label{UniqueNh}
Let $n\geq  n(e,\ell,b)$, $b\geq\max\{3t,4e+4\}$, $\ell\geq2$, $3\leq h\leq\ell$, then $N'(n,\ell,e,h)>N'(n,\ell,e,h+1)$ for each $h$ and thus, each $N_h$ exists and attains unique value.
\end{theorem}
\begin{proof}
Since $N'(n,\ell,e,h+1)$ denotes the maximum value for $|\bigcap_{i=1}^{h+1}B_t(\bc_i)|$ over all sets  $T(Y)=\{\bc_1,\dots,\bc_{h+1}\}$, we clearly have $N'(n,\ell,e,h)\geq N'(n,\ell,e,h+1)$. As we have seen in the previous theorems, the maximum value for $N'(n,\ell,e,h+1)$ is attained when the codewords in $T(Y)=\{\bc_1,\dots,\bc_{h+1}\}$ have supports $\supp(\bc_i)=[(i-1)(e+1)+1,i(e+1)]$. We show that when we choose $h$ of these codewords, then we can fit more words in the intersection of their $t$-balls. We clearly have $ \bigcap_{i=1}^{h+1}B_t(\bc_i)\subseteq \bigcap_{i=1}^{h}B_t(\bc_i)$. In the following, we show that there exists a word $\by\in \bigcap_{i=1}^{h}B_t(\bc_i)\setminus \bigcap_{i=1}^{h+1}B_t(\bc_i)$ and hence, $|\bigcap_{i=1}^{h}B_t(\bc_i)|>|\bigcap_{i=1}^{h+1}B_t(\bc_i)|$.

Let $w(\by)=\ell+1$, $|\supp(\by)\cap \supp(\bc_i)|=1$ for $i\in[1,h]$, $\supp(\by)\cap \supp(\bc_{h+1})=\emptyset$ and $|\supp(\by)\cap [(h+1)(e+1)+1,n]|=\ell+1-h$. We have $d(\by,\bc_i)=\ell+1+e+1-2=t$ for $i\in[1,h]$ but $d(\by,\bc_{h+1})=\ell+1+e+1=t+2$. Thus, the claim follows.
\end{proof}



Using Theorem \ref{KanavalukemaEkviv}, we can improve the bound $\LL\le \ell+1$ of Theorem \ref{l+1Theorem POL} just by adding a constant number $(e+1)^{\ell+1}$ of channels. 
\begin{corollary}\label{Cor:lLista}
Let $n\geq  n(e,\ell,b)$, $b\geq\max\{3t,4e+4\}$ and $\ell\geq3$. If $N\geq V(n,\ell-1)+(e+1)^{\ell+1}+1$, then $\LL\leq \ell$.\end{corollary}
\begin{proof}
We calculate the value $N_h$ with $h=\ell+1$ of Theorem~\ref{KanavalukemaEkviv}. For this purpose, we consider the set $W_w$ with $w\geq\ell$. We get that $i_j\geq 1$ for each $j\in[1,h]$ and thus, $w\geq \ell+1$. On the other hand, $w\leq \ell+1$. Indeed, if $w \geq \ell+2$, then $w\geq\sum_{j=1}^h i_j\geq \sum_{j=1}^{\ell+1} i_j \geq (\ell+1)(w+1-\ell)/2 = (\ell-1)(w+1-\ell)/2 + w -(\ell-1) > w$ (a contradiction). Therefore, as $w = \ell +1$ and $i_j \geq 1$ for each $j$, we have $W_w=\{(1,1,...,1)\}$. Thus, the sum corresponding to $W_w$ in Theorem \ref{KanavalukemaEkviv} gives $V(n,\ell-1)+(e+1)^{\ell+1}$. Hence, if $N \geq N_{\ell+1} + 1 = V(n,\ell-1)+(e+1)^{\ell+1} + 1$, then $\LL \leq \ell$.
\end{proof}



In the following corollary, we present the asymptotic behaviour of $N_h$ on $n$. Notice that Corollary \ref{Cor:lLista} considers the case $h=\ell+1$ and Corollary \ref{ListAsymptotics} cases $3\le h\le \ell$.    

\begin{corollary}\label{ListAsymptotics}
Let $\ell\geq h\geq3$, $n\geq  n(e,\ell,b)$ and $b\geq\max\{3t,4e+4\}$. If $N\geq N_h+1$, then $\LL\leq h-1$, where \begin{align*}
 N_h\in V(n,\ell-1)+\binom{n-h(e+1)}{\ell+1-h}(e+1)^{h}+\Theta(n^{\ell-h}).
\end{align*}\end{corollary}

\begin{proof}
Let $n\geq  n(e,\ell,b)$, $b\geq\max\{3t,4e+4\}$, $\ell\geq h\geq3$, and $e$ and $\ell$ be fixed. By Theorem \ref{KanavalukemaEkviv}, we have $N_h= V(n,\ell-1)+ \sum_{w\geq\ell}\sum_{(i_1,\dots,i_h)\in W_w}\binom{n-h(e+1)}{w-\sum_{j=1}^h i_j}\prod_{j=1}^h\binom{e+1}{i_j}$.  


When we inspect the set $W_w$ closer, we notice  that in $\binom{n-h(e+1)}{w-\sum_{j=1}^h i_j}$ term $w-\sum_{j=1}^h i_j$ attains its maximum value when $w-\sum_{j=1}^h i_j$ is as large as possible since $w\leq 2e+\ell+1$ (recall that $e$ and $\ell$ are constants with respect to $n$). This occurs exactly when $i_j=\lceil(w+1-\ell)/2\rceil$ for each $j$ and $\lceil(w+1-\ell)/2\rceil$ is as small as possible. In particular, when $w\in\{\ell,\ell+1\}$, we have $\lceil(w+1-\ell)/2\rceil=1$ and when $w=\ell+1+a$ for some integer $a\geq 1$, we have $\lceil(w+1-\ell)/2\rceil=1+\lceil a/2\rceil$. Moreover, when $w=\ell$, we have $w-\sum_{j=1}^h i_j\leq \ell-h$. When $w=\ell+1$, we have $w-\sum_{j=1}^h i_j\leq \ell+1-h$ and when $w=\ell+1+a$, we have $w-\sum_{j=1}^h i_j\leq \ell+1+a-h(1+\lceil a/2\rceil)$. Since $h\geq3$ and $a\geq1$, we have $\ell+1+a-h(1+\lceil a/2\rceil)\leq \ell+1-h-\lceil a/2\rceil\leq\ell-h$. Thus, we may concentrate on the case with $w=\ell+1$.

Furthermore, as $w = \ell+1$ and $i_j = 1$ for each $j = 1, 2, \ldots, h$, we have  $\prod_{j=1}^h\binom{e+1}{i_j}=(e+1)^h$. Recall that $e$ is constant on $n$. Hence, for large $n$, it is enough to consider the binomial coefficient $\binom{n-h(e+1)}{\ell+1-h}$. Moreover, the second largest binomial coefficient is $\binom{n-h(e+1)}{\ell-h}$ and we have $\binom{n-h(e+1)}{t}\in \Theta (n^t)$, when $n$ is large.  Therefore, we have $N_h\in  V(n,\ell-1)+\binom{n-h(e+1)}{\ell+1-h}(e+1)^h+\Theta(n^{\ell-h})$.\end{proof}

In Theorem \ref{YB raja l=2}, a tight bound for the number of channels to certainly attain the list size $\LL\leq2$ is presented when code $C$ has minimum distance $d$. 
Observe, that when we choose $h=3$ in Theorem \ref{KanavalukemaEkviv}, we attain the number of channels required to have $\LL\leq2$. The bound of Theorem \ref{YB raja l=2} by Yaakobi and Bruck looks quite different from the bound in Theorem \ref{KanavalukemaEkviv}. However, Theorem \ref{YB raja l=2} can be obtained as a corollary from Theorem \ref{KanavalukemaEkviv} as shown in Corollary \ref{YB raja e-error l=2COR}. The new presentation in Corollary \ref{YB raja e-error l=2COR} somewhat simplifies the inequalities for the indices compared to Theorem \ref{YB raja l=2}.

%
%
%

\begin{corollary}\label{YB raja e-error l=2COR}
Let $n\geq  n(e,\ell,b)$, $b\geq\max\{3t,4e+4\}$ and $\ell\geq2$. If $$N\geq N'=\sum_{i_1,i_2,i_3,i_4}\binom{n-3e-3}{i_1}\binom{e+1}{i_2}\binom{e+1}{i_3}\binom{e+1}{i_4}+1$$ for \begin{itemize}
\item $0\leq i_1\leq \ell-1$,
\item $0\leq i_4\leq \ell-1-i_1$, 
\item $0\leq i_3\leq \ell-1-i_1$ and 
\item $i_1+i_3+i_4-(\ell-1)\leq i_2\leq \ell-1-i_1-|i_4-i_3|$,
\end{itemize}  then $\mathcal{L}\leq2$ for any $e$-error-correcting code $C$ and value $N'$ is equal to the lower bound obtained in Theorem \ref{YB raja l=2}.
\end{corollary}
\begin{proof}
We first show that the formulation of the bound follows from Theorem \ref{KanavalukemaEkviv} when $h=3$. 
This gives the minimum number of channels $N$ required to have $\LL\leq2$. In particular, we have 


\begin{align*}
N_3=&V(n,\ell-1)+\sum_{w\geq\ell}\sum_{(i_2,i_3,i_4)\in W_w}\binom{n-3e-3}{w-i_2-i_3-i_4}
\binom{e+1}{i_2}\binom{e+1}{i_3}\binom{e+1}{i_4}.
\end{align*}
We have renamed the indices for convenience (the index $i_1$ will be saved for later use in the proof). Moreover, we have $W_w=\{(i_2,i_3,i_4)\mid (w+1-\ell)/2\leq i_j\leq e+1, w\geq i_2+i_3+i_4\}$. 
Earlier, we have used $W_w$ only with the assumption that $w\geq\ell$. However, we may allow here that $w < \ell$. Hence, we may have some binomial coefficients with $i_j<0$ for some $j$. In these cases (and when $i_j>e+1$), we use the common convention that the binomial coefficient attains the value $0$. Observe that $N_3=\sum_{w\geq0}\sum_{(i_2,i_3,i_4)\in W_w}\binom{n-3e-3}{w-i_2-i_3-i_4}\binom{e+1}{i_2}\binom{e+1}{i_3}\binom{e+1}{i_4}$. Indeed, since $W_w=\{(i_2,i_3,i_4)\mid 0\leq i_j\leq e+1, w\geq i_2+i_3+i_4\}$, we have $\sum_{w=0}^{\ell-1}\sum_{(i_2,i_3,i_4)\in W_w}\binom{n-3e-3}{w-i_2-i_3-i_4}\binom{e+1}{i_2}\binom{e+1}{i_3}\binom{e+1}{i_4} = \sum_{w=0}^{\ell-1}\sum_{i_2=0}^{e+1}\sum_{i_3=0}^{e+1}\sum_{i_4=0}^{e+1} \binom{n-3e-3}{w-i_2-i_3-i_4}\binom{e+1}{i_2}\binom{e+1}{i_3}\binom{e+1}{i_4} = V(n,\ell-1)$ due to binomial identity $\sum_{p=0}^s\binom{m-s}{k-p}\binom{s}{p}=\binom{m}{k}.$ 


Let us denote by $i_1=w-i_2-i_3-i_4$.  We now get that $2i_2 \geq 2(w+1-\ell)/2 = i_1+i_2+i_3+i_4-\ell+1$ and similar inequalities for $i_3$ and $i_4$. Since we do not have to take into account lower bound $i_1\geq0$ (cases with $i_1<0$ increase binomial sum by $0$) or the cases with $i_j>e+1$ for $j\in\{2,3,4\}$, we can consider following system of inequalities:
\begin{align}
 i_2\geq& i_1+i_3+i_4-\ell+1\label{i_2}\\
i_3\geq& i_1+i_2+i_4-\ell+1\label{i_3}\\
 i_4\geq& i_1+i_2+i_3-\ell+1\label{i_4}. 
 \end{align} 
Our goal is to show that this system of inequalities is equivalent with the following system of inequalities: 
\begin{align}
i_4&\leq \ell-1-i_1, \label{i'_4}\\
i_3&\leq \ell-1-i_1,  \label{i'_3}\\
i_1+i_3+i_4-(\ell-1)\leq i_2&\leq \ell-1-i_1-|i_4-i_3|.\label{i'_2}
\end{align}

Let us first show that the second system of inequalities follows from the first system of inequalities.

Inequality (\ref{i'_4}) follows from $$i_4=w-i_1-i_2-i_3\leq w-i_1-2(w-\ell+1)/2=\ell-1-i_1.$$ We obtain Inequality (\ref{i'_3}) in similar manner. Moreover, from Inequalities (\ref{i_3}) and (\ref{i_4}) we obtain $i_2\leq \ell-1-i_1-i_4+i_3$ and $i_2\leq \ell-1-i_1-i_3+i_4$, respectively. Together, these imply $$i_2\leq \ell-1-i_1-|i_4-i_3|.$$ Finally, the lower bound inequality in (\ref{i'_2}) follows directly from (\ref{i_2}).

Let us then show that the first system of inequalities follows from the second one. First of all, Inequality (\ref{i_2}) follows immediately from Inequality (\ref{i'_2}). Assume first that $i_4\geq i_3$. Then the upper bound of Inequality (\ref{i'_2}) is $i_2\leq \ell-1-i_1-i_4+i_3$. This implies Inequality (\ref{i_3}) and inequality (\ref{i_4}) since $$i_4\geq i_3\geq i_1+i_2+i_4-\ell+1\geq i_1+i_2+i_3-\ell+1.$$
The case with $i_3\geq i_4$ is similar.

Finally, we may add lower bounds $i_j\geq0$ for all $j\in\{1,2,3,4\}$ due to binomial coefficient context. Similarly we notice that if $i_1\geq\ell$, then $i_4<0$. Thus, we may also add upper bound  $i_1\leq \ell-1$. Hence, the first part of the claim follows.

Let us then derive the bound of Theorem \ref{YB raja l=2} by Yaakobi and Bruck from this new lower bound. The case with $d=2e+1$ is included in Appendix and we consider here only the case with $d=2e+2$. 
When we have $d=2e+2$, Theorem \ref{YB raja l=2} can be presented in the following way: If
\noindent$N\geq\sum_{h_1,h_2,h_3,h_4}\binom{n-3e-3}{h_1}\binom{e+1}{h_2}\binom{e+1}{h_3}\binom{e+1}{h_4}+1$ for \begin{itemize}
\item $0\leq h_1\leq \ell-1$,
\item $h_1-(\ell-1)\leq h_4\leq \ell-1-h_1$, 
\item $e+2-\ell+h_1\leq h_3\leq t-(h_1+h_4)$ and 
\item $\max\{h_1-h_3-h_4+2e+3-\ell,h_1+h_3+h_4-\ell+1\}\leq h_2\leq \ell-1-(h_1+h_4-h_3)$,
\end{itemize}  then $\mathcal{L}\leq2$ for any $e$-error-correcting code $C$ (with minimum distance $d=2e+2$). Next, we modify the presentation we got for $N_3$ into the formulation above.

Let us denote by $i'_2=e+1-i_2$ and by $i'_3=e+1-i_3$. Observe that $\binom{e+1}{i'_j}=\binom{e+1}{i_j}$ for $j\in\{2,3\}$. We can replace lower bound $i_4\geq0$ by $i_4\geq i_1-(\ell-1)$ since $i_4\geq0\geq i_1-\ell+1$  and $\binom{e+1}{i_4}=0$ when $i_4<0$. Moreover, we have $0\leq i_3\leq \ell-1-i_1$. Hence, $e+1\geq i'_3\geq e+2-\ell + i_1$. Notice that the upper bound of $i'_3$ can be replaced by $t-(i_1+i_4)$ since $t-(i_1+i_4)\geq e+1$ as $i_1+i_4\leq \ell-1$ and $\binom{e+1}{i'_3}=0$ when $i'_3>e+1$.  

For $i_2$ we have $i_1+i_3+i_4-(\ell-1)\leq i_2\leq \ell-1-i_1-|i_4-i_3|$ and hence, $\ell-1-(i_1+i_4-i'_3)\geq i'_2\geq -\ell+e+2+i_1+|i_4-i_3|=e+2-\ell+i_1+|i_4+i'_3-e-1|=\max\{i_1-i'_3-i_4+2e+3-\ell,i_1+i'_3+i_4-\ell+1\}$. By comparing these inequalities with the bounds used in Theorem \ref{YB raja l=2}, we notice that they are identical.
 The case with $d=2e+1$ is similar and is included in Appendix. Hence, we get the claim.\end{proof}





\section{New Bounds with the aid of Covering Codes} \label{SectionCoveringCodes}

Notice that although we have the bound $\LL\le \ell+1$ when $n$ is rather large (see Theorem~\ref{l+1Theorem POL}), for smaller lengths of the codes our best bound is still $\LL\le 2^\ell$ (see Theorem~\ref{shatter raja}) when the number of channels satisfies $N\ge V(n,\ell-1)+1$. Although this bound is attained in some cases (see \cite{junnila2020levenshtein}) and thus cannot be improved in general, we can try to get a smaller list size $\LL$ when we increase the number of channels as we have seen Theorem \ref{KanavalukemaUusi}.  In this section, we utilize covering codes when we increase the number of channels. A code $C\subseteq \F^n$ is an $R$-\emph{covering code} if for every word $\bx\in \F^n$ there exists a codeword $\bc\in C$ such that $d(\bx,\bc)\le R$. For an excellent source on results concerning covering codes, see \cite{chll}. Let us denote by $k[n,R]$ the smallest possible dimension of a \emph{linear} $R$-covering code of length $n$.

Let us next present the well-known Sauer-Shelah lemma (see~\cite{sauer1972density, shelah1972combinatorial}). Let $\mathcal{F}$ be a family of subsets of $[1,n]$, where $n$ is a positive integer. We say that a subset $S$ of $[1,n]$ is \emph{shattered} by $\mathcal{F}$ if for any subset $E \subseteq S$ there exists a set $F \in \mathcal{F}$ such that $F \cap S = E$. The Sauer-Shelah lemma states that if $|\mathcal{F}| > \sum_{i=0}^{k-1} \binom{n}{i}$, then $\mathcal{F}$ shatters a subset of size (at least) $k$. Since the subsets of $[1,n]$ can naturally be interpreted as words of $\F^n$, the Sauer-Shelah lemma can be reformulated as follows. Notice that $\sum_{i=0}^{k-1} \binom{n}{i}=V(n,k-1)$.
\begin{theorem}[\cite{sauer1972density, shelah1972combinatorial}]\label{Sauer-Shelah}
	If $Y \subseteq \F^n$ is a set containing at least $V(n,k-1)+1$ words, then there exists a set $S$ of $k$ coordinates such that for any word $\bw \in \F^n$ with $\supp(\bw) \subseteq S$ there exists a word $\mathbf{s} \in Y$ satisfying $\supp(\bw) = \supp(\mathbf{s}) \cap S$. Here we say that the set $S$ of coordinates is \emph{shattered} by $Y$.
\end{theorem}

Observe  that each Hamming ball of radius $e$ contains at most one codeword of $C$. Thus, if the intersection of the balls of radius $t$ centered at the output words of $Y$ can be covered by $k$ balls of radius $e$, then we have $|T(Y)|\leq k$. This approach is formulated in the following lemma.
\begin{lemma}[\cite{junnila2020levenshtein}]\label{e-pallot}
	Let $C\subseteq \F^n$ be an $e$-error-correcting code. If for any set of output words $Y = \{\y_1, \ldots, \y_N\}$ we have $$T(Y)\subseteq \bigcup_{i=1}^k B_e(\beta_i)$$ for some words $\beta_i \in \F^n$ ($i=1, \ldots, k$), then $\LL \leq k.$
\end{lemma}

Notice that Lemma \ref{e-pallot} also gives a decoding algorithm. Indeed, if the words $\beta_i$ are known, then there is at most one codeword in each $B_e(\beta_i)$, we can use the decoding algorithm of $C$ on $\beta_i$ and the codeword can be added to the list $T$.

\begin{theorem}
	Let $C$ be an $e$-error-correcting code. If the number of channels satisfies $N\ge V(n,\ell+2R-1)-2^{\ell+2R-k[\ell+2R,R]}+2$, then
	$$\LL\le  2^{k[\ell+2R,R]}.$$
\end{theorem}

\begin{proof}
Let $\bx$ be the input word. We have $|Y|\ge  (V(n,\ell+2R-1)+1)-(2^{\ell+2R-k[\ell+2R,R]}-1)$. Next we show that with this number of outputs we can guarantee that there exists a set $S$ of $\ell+2R$ coordinates such that within these coordinates of $S$ a subset $Y'\subseteq Y$ contains a linear $R$-covering code of length $\ell+2R$.  Due to Theorem \ref{Sauer-Shelah}, we know that if we had more output words, namely, $|Y|\ge V(n,\ell+2R-1)+1$, then we would have a set $S$ of coordinates such that a subset $Y''\subseteq Y$ contains all the $2^{\ell+2R}$  words of length $\ell+2R$ among these coordinates of $S$. Let $D$ be a linear $R$-covering code in   $\F^{\ell+2R}$ with $\dim(D)=k[\ell+2R,R].$ Notice that any coset $\bu+D$, $\bu\in \F^{\ell+2R}$, of the linear code $D$ is also an $R$-covering code, and there are $2^{\ell+2R-\dim(D)}$ distinct cosets. Therefore, the set $Y''$ can miss any $2^{\ell+2R-\dim(D)}-1$ words of $\F^{\ell+2R}$ and still the remaining subset contains at least one $R$-covering code of length  $\ell+2R$. Consequently, it follows that $Y'$ contains an $R$-covering code of size $2^{k[\ell+2R,R]}$ because $Y'$ can be obtained from $Y''$ by removing some $2^{\ell+2R-\dim(D)}-1$ words.

Now let $\s\in \F^n$ be a word such that $\supp(\s)=S$ and $Y_1=\{\by_1,\dots, \by_{2^{k[\ell+2R,R]}}\}\subseteq Y'$ the subset of output words corresponding to the $R$-covering code. Denote $\beta_i=\s+\by_i$ for $i=1,\dots, 2^{k[\ell+2R,R]}$. Since the words in set $Y_1$ form, among the coordinates corresponding to $S$, an $R$-covering code  of length  $\ell+2R$, we know that there exists $\by_j$, $j\in\{1,\dots,2^{k[\ell+2R,R]}\}$, such that  the words $\by_j$ and $\bx+\s$ differ in at most $R$ places among the coordinates of $S$. Consequently, as $d(\bx,\by_j)\le t$, the words $\bx$ and $\beta_j=\by_j+\s$ have distance at most $t-(\ell+R)+R=e$ from one another. Therefore, by Lemma~\ref{e-pallot}, we get that 
$\LL\le 2^{k[\ell+2R,R]}.$
 \end{proof}

Note that if $\ell=5$ and $N\ge V(n,4)+1$, then, by Theorem~\ref{shatter raja}, we have $\LL\le 2^5=32.$ If we have $N\ge  V(n,6)-6$, then (using as the linear $1$-covering code $D$ the Hamming code of length 7), we obtain by the previous result, that $\LL\le 16.$

\section{List size with less channels} \label{SectionLessChannels}

By the following theorem (of~\cite{junnila2020levenshtein}), it is clear that if we have \emph{less} than $V(n,\ell-1)+1$ channels, then the list size cannot in general be constant for  $e$-error-correcting codes of length $n$. 

\begin{theorem} \cite{junnila2020levenshtein}\label{Asymptotic L}
	Let $V(n,\ell-p-1)+1\leq N\leq V(n,\ell-p)$ where $0\leq p\leq \ell-1$. Moreover, let $C\subseteq \F^n$ be such an $e$-error-correcting code that $\mathcal{L}$ is maximal. Then we have $$\mathcal{L}=\Theta(n^p).$$
\end{theorem}	

Due to this result, in order to have a smaller list size, let us concentrate on certain $e$-error-correcting codes, namely, those with at most $M$ codewords within any ball of radius $e+a$, for some $a>0$.

\begin{theorem} \label{tldc}
	Let $N\ge V(n,\ell-a-1)+1$ where  $0\leq a\leq \ell-1$. Let $C$ be an $e$-error-correcting code such that $|B_{e+a}(\bu)\cap C|\le M$ for every $\bu\in \F^n$. Consequently,
	$$\LL\le 2^{\ell-a}M.$$ 
\end{theorem}
\begin{proof}
Assume that we received the set $Y$ of output words from the channels. Due to the number of channels, we know, by Theorem~\ref{Sauer-Shelah}, that there exists a subset $Y'\subseteq Y$ of size $|Y'|=2^{\ell-a}$ such that these words have in some $\ell-a$ coordinates, denote this set of coordinates by $S$, all possible words of length $\ell-a$.  
Suppose $\bx$ is the input word and denote $Y'=\{\by_1,\dots,\by_{2^{\ell-a}}\}.$ 
Let $\beta_i=\by_i+\s$, $i=1,\dots, 2^{\ell-a}$ where $\supp(\s)=S$. It is easy to check that $\bx$ has distance at most $e+a$ from one of the $\beta_i$'s. Since there are at most $2^{\ell-a}$ $\beta_i$'s and there are at most $M$ codewords within distance $e+a$ from each of them, we obtain the bound $\LL\le 2^{\ell-a}M.$
\end{proof}

The previous result is useful when our $e$-error-correcting code is a code for traditional list-decoding, see \cite{Guruswamin_kirja}. For number of channels being less than $V(n,\ell-1)+1$, it also gives, for every $e$-error-correcting code with suitable $a$, small exponent for $n$ compared to Theorem~\ref{Asymptotic L} (see Corollary~\ref{Johnson}(ii) below), or even constant bounds (see Corollary~\ref{Johnson}(i)).
The following corollary follows straightforwardly from  applying the results in \cite[Theorem 3.2]{Guruswamin_kirja} to Theorem~\ref{tldc}. Let us denote $$r(n,e,M)=\frac{n}{2}\left(1-\sqrt{1-\frac{M-1}{M}\frac{2(2e+1)}{n}}\right)$$
and 
$$r(n,e)=\frac{n}{2}\left(1-\sqrt{1-\frac{2(2e+1)}{n}}\right).$$

\begin{corollary}\label{Johnson} Let $C$ be an $e$-error-correcting code, $t=e+\ell$, an integer $M\ge 1$ and $2e+1<n/2$. We have
	\begin{itemize}
\item[(i)]  Let $N\ge V(n,\ell-r(n,e,M)+e-1)+1$ where  $0\leq r(n,e,M)-e\leq \ell-1$ and $r(n,e,M) \geq 1$. Consequently,
	$$\LL\le 2^{t-r(n,e,M)}M.$$ 

\item[(ii)] Let $N\ge V(n,\ell-r(n,e)+e)+1$ where  $1\leq r(n,e)-e\leq \ell$ and $r(n,e) \geq 2$. Consequently,
	$$\LL\le 2^{t-r(n,e)+1} n.$$ 
	\end{itemize}
\end{corollary}

In what follows, we continue our study of the list size for codes with $|B_{e+a}(\bu)\cap C|\le M$ for every $\bu\in \F^n$.
First, we introduce two technical lemmas. Lemma \ref{l-a-1Patka} can be seen as a reformulation of \cite[Lemma $13$]{junnila2020levenshtein} for lesser number of channels and Lemma \ref{dist2l-2a-2} as a reformulation of Lemma \ref{Listasanat lahella POL}.
 
\begin{lemma}\label{l-a-1Patka}
Let $N\geq V(n,\ell-a-1)+1$, $\ell-1\geq a\geq1$ and $n\geq (\ell-a-1)^22^b+\ell-a-2$. Then for any codeword $\bc\in T(Y)$ and any set  $\overline{D}\subseteq [1,n]$ with cardinality $|\overline{D}|=b$ there exists an output word $\by\in Y$ such that $|\supp(\by+\bc)\setminus \overline{D}|\geq \ell-a-1$.
\end{lemma}
\begin{proof}
Let us assume without loss of generality that $\bc=\nolla$. Thus, $w(\by)\leq t$ for each $\by\in Y$. Let us count the number of binary words of weight at most $t$ which have $|\supp(\by)\setminus \overline{D}|< \ell-a-1$. There are at most \begin{align*}
&\sum_{j=0}^{\ell-a-2}\sum_{i=0}^{\min\{b,t-j\}}\binom{b}{i}\binom{n-b}{j}\\
\leq&(\ell-a-1)2^b\binom{n}{\ell-a-2}\\
=&(\ell-a-1)2^b\frac{\ell-a-1}{n-\ell+a+2}\binom{n}{\ell-a-1}\\
\leq& \binom{n}{\ell-a-1}
\end{align*}
such words, which is strictly less than $N$ and hence, the claim holds.
\end{proof}

\begin{lemma}\label{dist2l-2a-2}
Let $N\geq V(n,\ell-a-1)+1$, $\ell-1\geq a\geq1$, $b=2t$, $n\geq (\ell-a-1)^22^b+\ell-a-2$ and $\LL\geq2$. Then for any two codewords $\bc_1$ and $\bc_2$ in $T(Y)$ holds $$2e+1\leq d(\bc_1,\bc_2)\leq 2e+2a+2.$$
\end{lemma}
\begin{proof}
Let us assume without loss of generality that $\bc_1=\nolla$. Notice that $w(\bc_2)\leq2t$. Choose $\overline{D}=\supp(\bc_2)$. Now, by Lemma \ref{l-a-1Patka}, there exists an output word $\by\in Y$ such that $|\supp(\by)\setminus \overline{D}|\geq\ell-a-1$. Since $d(\bc_2,\by)\leq t$ and $w(\by)\leq t$, we have $t\geq \ell-a-1+(w(\bc_2)-(t-(\ell-a-1)))=\ell-2a-2-e+w(\bc_2)$ and hence, $w(\bc_2)\leq 2e+2a+2$. Therefore, $d(\bc_1,\bc_2)\leq 2e+2a+2$. Since we consider an $e$-error-correcting code, the lower bound follows.
\end{proof}

Together with the two previous lemmas, we can now prove an upper bound for $\LL$, depending on the maximum number $M$ of codewords in an $(e+a)$-radius ball.

\begin{theorem}\label{list(t+2)M}
Let $N\geq V(n,\ell-a-1)+1$, $\ell-1\geq a\geq1$, $b=\left\lceil(2e+2a+2)^{\mathbbm{e}\cdot(e+a+1)!}\right\rceil$, $n\geq (\ell-a-1)^22^b+\ell-a-2$. Moreover, let $C$ be such an $e$-error-correcting code that there are at most $M$ codewords in any $(e+a)$-radius ball. Then $$\LL\leq \max\{(t+1)M,b/(2e+2a+2)\}.$$
\end{theorem}
\begin{proof}
Let us denote $T(Y)=\{\bc_1,\bc_2,\dots,\bc_\LL\}$. If $\LL\leq b/(2e+2a+2)$, then the claim follows. Assume now that $\LL> b/(2e+2a+2)$. Assume then, without loss of generality, that $\bc_1=\nolla$. By Lemma \ref{dist2l-2a-2}, we have $w(\bc_i)\leq 2e+2a+2$ for each $i\in[1,\LL]$. 

Let $Z$ be a subset of $\F^n$. We say that $\bw \in \F^n$ is a \emph{central word} with respect to $Z$ if $d(\bw, \bz) \leq e+a+1$ for all $\bz \in Z$. Moreover, the set of all central words with respect to $Z$ is denoted by $W_Z$. In what follows, we first show a useful observation stating that if a subset $C_S \subseteq T(Y)$ is such that $\nolla \in C_S$ and $|C_S| \leq b/(2e+2a+2)$, then there exists a word $\bw \in \F^n$ such that $w(\bw) \leq e+a+1$ and $d(\bw, \bc) \leq e+a+1$ for any $\bc \in C_S$, i.e., $\bw \in W_{C_S} \neq \emptyset$. Since $w(\bc) \leq 2e+2a+2$ for any $\bc  \in C_S$, we have
\[
\left| \bigcup_{\bc\in C_S} \supp(\bc)\right| \leq (2e+2a+2) \cdot \frac{b}{2e+2a+2} = b \text.
\]
Therefore, by Lemma~\ref{l-a-1Patka}, there exists an output word $\by\in Y$ such that 
\[
|\supp(\by)\setminus \bigcup_{\bc\in C_S}\supp(\bc)|\geq \ell-1-a \text.
\]
Let then $\bw \in \F^n$ be such that
\[
\supp(\bw) = \supp(\by) \cap \bigcup_{\bc\in C_S}\supp(\bc) \text.
\]
Now, as $d(\by, \bc) \leq t = e + \ell$ for any $\bc \in C_S$, we have $d(\bw, \bc) \leq e+a+1$. Moreover, $w(\bw) \leq e+a+1$ since $\nolla \in C_S$. 


 
In what follows, we show using an iterative approach that there exists a central word $\bw \in \F^n$ with respect to $T(Y)$. We begin the iterative process by considering a subset $C_0 = \{\bc_1,\bc_2\} \subseteq T(Y)$ such that $\bc_1 = \nolla$ and $w(\bc_2) = e+a+1+p_1$ with $1 \leq p_1 \leq e+a+1$. Indeed, we may assume that such a codeword $\bc_2$ exists since otherwise we are immediately done due to $\bw = \nolla$ being the searched central word. Observe that the weight of a central word $\bw \in W_{C_0}$ satisfies $p_1\leq w(\bw)\leq e+a+1$. In particular, there are exactly $\binom{e+a+1+p_1}{p_1}$ central words $\bw\in W_{C_0}$ of weight $p_1$. For each such central word $\bw$, we may assume that there exists a codeword $\bc \in T(Y)$ such that $d(\bw, \bc) > e+a+1$ as otherwise $\bw \in W_{T(Y)}$ and we are done. Now we form a new code $C_1$ by adding such a codeword $\bc$ for each $\bw \in W_{C_0}$ with $w(\bw) = p_1$. The number of added codewords is at most
\[
\binom{e+a+1+p_1}{p_1}\leq (2e+2a+2)^{p_1}-1 \text.
\]
Therefore, we have $|C_1| \leq (2e+2a+2)^{p_1}+1$. Notice that there are no central words in $W_{C_1}$ of weight at most $p_1$. Furthermore, by the previous observation for $W_{C_S}$, $W_{C_1}$ is nonempty as $|C_1| \leq (2e+2a+2)^{p_1}+1 \leq b/(2e+2a+2)$.


Assume that $p_2>p_1$ is now the smallest weight of a central word in $W_{C_1}$. Let $\bw$ be a central word with respect to $C_1$ of weight $p_2$. The support of $\bw$ is a subset of $\bigcup_{\bc\in C_1}\supp(\bc)$ since otherwise there exists a central word $\bw'\in W_{C_1}$ with $w(\bw')<w(\bw)$ (a contradiction). Therefore, as ($\bc_1 = \nolla \in C_1$ implies) $\left| \bigcup_{\bc\in C_1}\supp(\bc)\right| \leq (2e+2a+2)(2e+2a+2)^{p_1}$, the number of central words of weight $p_2$ in $W_{C_1}$ is at most
\begin{align*}
	&\binom{(2e+2a+2)(2e+2a+2)^{p_1}}{p_2}\\
	=&\binom{(2e+2a+2)^{p_1+1}}{p_2}\\
	\leq& \frac{(2e+2a+2)^{p_1p_2+p_2}}{2} \text.
\end{align*}
Again, for each central word $\bw \in W_{C_1}$ of weight $p_2$, there exists a codeword $\bc \in T(Y)$ such that $d(\bw, \bc) > e+a+1$. Now we form a new code $C_2$ by adding such a codeword $\bc$ for each $\bw \in W_{C_1}$ with $w(\bw) = p_2$. Thus, we have $|C_2|\leq (2e+2a+2)^{p_1}+1+(2e+2a+2)^{p_1p_2+p_2}/2\leq (2e+2a+2)^{p_1p_2+p_2}/2+(2e+2a+2)^{p_1p_2+p_2}/2=(2e+2a+2)^{p_1p_2+p_2}$. 

The process can be iteratively continued by forming a new code $C_i$ based on the previous code $C_{i-1}$, until we have reached $p_i = e+a+1$ or we have already found a central word with respect to $T(Y)$. In what follows, the iterative process is explained in more detail:
\begin{itemize}
	\item Assume that $p_{i}>p_{i-1}$ is the smallest weight of a central word in $W_{C_{i-1}}$.
	\item Assume that  $|C_{i-1}|\leq (2e+2a+2)^{\sum_{h=1}^{i-1}\prod_{k=h}^{i-1} p_k}$. By the observation above for $W_{C_S}$, this implies that $W_{C_{i-1}}$ is nonempty as $|C_{i-1}|\leq b/(2e+2a+2)$ (see also Equation~\eqref{Cj_upper_bound}). Furthermore, we have $|\bigcup_{\bc\in C_{i-1}}\supp(\bc)|\leq (2e+2a+2)|C_{i-1}| = (2e+2a+2)^{(\sum_{h=1}^{i-1}\prod_{k=h}^{i-1} p_k)+1}$.
	\item Let $\bw \in W_{C_{i-1}}$ be of weight $p_i$. As previously, the support of $\bw$ is a subset of $\bigcup_{\bc\in C_{i-1}}\supp(\bc)$ since otherwise there exists a central word $\bw'\in W_{C_{i-1}}$ with $w(\bw')<w(\bw)$ (a contradiction). Therefore, as $|\bigcup_{\bc\in C_{i-1}}\supp(\bc)|\leq (2e+2a+2)|C_{i-1}| = (2e+2a+2)^{(\sum_{h=1}^{i-1}\prod_{k=h}^{i-1} p_k)+1}$, the number of central words in $W_{C_{i-1}}$ of weight $p_i$ is at most $\binom{(2e+2a+2)^{(\sum_{h=1}^{i-1}\prod_{k=h}^{i-1} p_k)+1}}{p_i}$.
	\item Again, for each central word $\bw \in W_{C_{i-1}}$ of weight $p_i$, there exists a codeword $\bc \in T(Y)$ such that $d(\bw, \bc) > e+a+1$. Now we form a new code $C_i$ by adding such a codeword $\bc$ for each $\bw \in W_{C_{i-1}}$ with $w(\bw) = p_i$. Thus, we have
	\begin{align}\label{C_iApprox} 
		|C_i| \leq& (2e+2a+2)^{\sum_{h=1}^{i-1}\prod_{k=h}^{i-1} p_k}
		%
		+ \binom{(2e+2a+2)^{1+(\sum_{h=1}^{i-1}\prod_{k=h}^{i-1} p_k)}}{p_i}\nonumber\\
		\leq& (2e+2a+2)^{\sum_{h=1}^{i}\prod_{k=h}^{i} p_k}/2
		%
		+(2e+2a+2)^{\sum_{h=1}^{i}\prod_{k=h}^{i} p_k}/2\\
		=&(2e+2a+2)^{\sum_{h=1}^{i}\prod_{k=h}^{i} p_k}\nonumber.
	\end{align} 
\end{itemize}
 
Notice that since $1\leq p_1< p_2<\cdots < p_i\leq e+a+1$, we reach $p_j = e+a+1$ at some point (or the central word $\bw$ with respect to $T(Y)$ has already been found in an earlier step). By~\eqref{C_iApprox}, we have 
\begin{equation} \label{Cj_upper_bound}
\begin{split}
	|C_j| &\leq (2e+2a+2)^{\sum_{h=1}^{j}\prod_{k=h}^{j} p_k} \\ &\leq(2e+2a+2)^{\sum_{h=1}^{e+a+1}\prod_{k=h}^{e+a+1} k} \\ &= (2e+2a+2)^{\sum_{h=0}^{e+a}(e+a+1)!/h!}	\\ &\leq (2e+2a+2)^{\mathbbm{e}\cdot(e+a+1)!-1}\\ 
	&\leq \frac{b}{2e+2a+2} 
\end{split}
\end{equation}
since $\sum_{h=0}^{e+a+1}1/h!< \sum_{h=0}^\infty1/h!=\mathbbm{e}$. Therefore, by the previous observation, $W_{C_j}$ is nonempty. Thus, in conclusion, there exists a central word $\bw \in \F^n$ with respect to $T(Y)$.

Let us now translate the Hamming space so that $\bw=\nolla$ and thus, $T(Y)\subseteq B_{e+a+1}(\nolla)$. In other words, we have  $w(\bc_i)\leq e+a+1$ for each $i$. Recall that we have $N\geq V(n,\ell-a-1)+1$. Thus, there exists a word $\by\in Y$ such that $w(\by)\geq \ell-a$. Moreover, since $d(\bc_i,\by)\leq t$ for each $i$, we have $w(\by)\leq t+e+a+1$. The proof now divides into the following three cases depending on the weight of $w(\by)$. 

(i) Assume first that $\ell-a \leq w(\by) \leq t$. Now the support of each $\bc_i \in T(Y)$ of weight $e+a+1$ intersects with $\supp(\by)$ since otherwise $d(\by,\bc_i) \geq (e+a+1) + (\ell-a) = t+1$ (a contradiction). Hence, 
\[
T(Y)\subseteq B_{e+a}(\nolla) \cup\bigcup_{i\in \supp(\by)}B_{e+a}(\e_i) \text.
\]

(ii) Assume then that $t \leq w(\by) \leq t+e+a$. Let $\by_s$ be a word such that $\supp(\by_s) \subseteq \supp(\by)$ and $w(\by_s) = t$. Now the support of each $\bc_i \in T(Y)$ of weight $e+a+1$ intersects with $\supp(\by_s)$ since otherwise $d(\by,\bc_i) \geq t+1$ (a contradiction). Hence, 
\[
T(Y)\subseteq B_{e+a}(\nolla) \cup\bigcup_{i\in \supp(\by_s)}B_{e+a}(\e_i) \text.
\]

(iii) Assume finally that $w(\by) = t+e+a+1$. Then we have $w(\bc_i) = e+a+1$ for any $\bc_i \in T(Y)$ as $d(\by,\bc_i) \leq t$. Let $\by_s$ be a word such that $\supp(\by_s) \subseteq \supp(\by)$ and $w(\by_s) = t+1$. Again the support of each $\bc_i \in T(Y)$ of weight $e+a+1$ intersects with $\supp(\by_s)$. Observing that $w(\bc_i) = e+a+1$ for any $\bc_i \in T(Y)$ as $d(\by,\bc_i) \leq t$, we have 
\[
T(Y)\subseteq \bigcup_{i\in \supp(\by_s)}B_{e+a}(\e_i) \text.
\]

Based on the cases~(i)--(iii), the set of codewords $T(Y)$ is always contained in a union of $t+1$ balls of radius $e+a$. Thus, as each ball of radius $e+a$ has at most $M$ codewords, we obtain that $\LL\leq (t+1)M$.
\end{proof}


\section{Decoding with majority algorithm}

In this section, we focus on decoding the transmitted word $\bx = (x_1, x_2, \ldots, x_n) \in C$ based on the set $Y$ of the output words using a majority algorithm. For the rest of the section, we assume that each word of $B_t(\bx)$ is outputted from a channel with equal probability. Here we actually allow --- unlike elsewhere in the paper --- some of the output words $\by_i$ to be equal. Probabilistic set-up has been studied for different error types, for example, in~\cite{batu2004reconstructing, viswanathan2008improved}. Our approach differs from these articles in that we have given an upper limit for the possible number of errors in any single channel. This allows us to have a verifiability property in Theorem \ref{Thm_majority_algorithm} unlike, for example, in~\cite{batu2004reconstructing, viswanathan2008improved}. With the verifiability property we mean that although we cannot be certain whether we can deduce the transmitted word correctly before looking at the output words, we can sometimes deduce the transmitted word with complete certainty after seeing the output words. In other words, some output word sets have properties which allow us to know the transmitted word with certainty. 

 First we describe the (well-known) majority algorithm using similar terminology and notation as in~\cite{yaakobi2018uncertainty}. The coordinates of the output words $\by_j \in Y$ are denoted by $\by_j = (y_{j,1}, y_{j,2}, \ldots, y_{j,n})$. Furthermore, the number of zeros and ones in the $i$th coordinates of the output words are respectively denoted by
\[
m_{i,0} = |\{j \in \{1, 2, \ldots, N\} \mid y_{j,i} = 0\}|
\]
and $m_{i,1} = N - m_{i,0}$.
Based on $Y$, the \emph{majority algorithm} outputs the word $\bz_Y = \bz = (z_1, z_2, \ldots, z_n) \in \F^n$, where
\[
z_i = 
\begin{cases}
	0 & \text{ if } m_{i,0} > m_{i,1} \\
	? & \text{ if } m_{i,0} = m_{i,1} \\
	1 & \text{ if } m_{i,0} < m_{i,1} \text.
\end{cases} 
\]
In other words, for each coordinate of $\bz$, we choose $?$, $0$ or $1$ based on whether the numbers of $0$s and $1$s are equal or which ones occur more frequently. Observe that the coordinate $z_i$ outputted by the majority algorithm is equal to $x_i$ if and only if at most $\lceil N/2 \rceil -1$ errors occur in the $i$th coordinates of $Y$. Observe that the complexity of the majority algorithm is $\Theta(Nn)$ and since reading all the output words takes $\Theta(Nn)$ time, majority algorithm has optimal time complexity.

In~\cite[Example~1]{yaakobi2018uncertainty}, it is shown that the majority algorithm does not always output the correct transmitted word $\bx$ when the number of channels $N$ is equal to the value in Theorem \ref{L=1} even though we take the $e$-error-correction capability of $C$ into account. In~\cite{Uusi_Maria_Abu-Sini}, a modification of the majority algorithm is presented for decoding and it is shown that if the number of channels satisfies the bound of Theorem~\ref{L=1}, then the output word of the algorithm belongs to $B_e(\bx)$ and can be uniquely decoded to $\bx$. In what follows, we demonstrate that \emph{with high probability} the word $\bz$ is within distance $e$ from $\bx$ with significantly smaller number of channels (than in~\cite{Uusi_Maria_Abu-Sini}). For example, the Monte Carlo simulations for the values $t=5$ and $n=28$ are illustrated in Table~\ref{Table_probabilities}.

For this purpose, we first consider a variant of the so called \emph{multiple birthday problem}; the multiple birthday problem has been studied, for example, in~\cite{Suzuki_et_al} and \cite{Kounavis_et_al}. Here we assume that $s$, $q$, $n$ and $t$ are integers satisfying $2 \leq s \leq q$ and $0 \leq t \leq n$. A \emph{throw} consists of placing $t$ balls randomly into $n$ buckets in such a way that each ball lands in a different bucket and each subset of the buckets of size $t$ for a throw has an equal probability. Denote then by $C_t(n,q,s)$ the event that after $q$ throws, at least one bucket contains at least $s$ balls. Observe that if $t=1$, then we are actually considering the multiple birthday problem; furthermore, if ($t =1$ and) $s=2$, then the case is the (regular) birthday problem. Furthermore, by $Pr[C_t(n,q,s)]$ we denote the probability that there exist at least $s$ balls in a bucket. In the following theorem, we present (based on $Pr[C_t(n,q,s)]$) a lower bound on the probability that the output $\bz$ of the majority algorithm is equal to $\bx$.
\begin{theorem} \label{Thm_probability_majority}
	Let $\bx$ be the transmitted codeword of $C \subseteq \F^n$ and $N$ the number of channels. The probability that the output $\bz$ of the majority algorithm is equal to $\bx$ is at least
	\[
	1-Pr[C_t(n,N,\lceil N/2 \rceil)] \text.
	\]
\end{theorem}
\begin{proof}
	Let $P_1$ denote the probability that a coordinate of the outputs $Y$ contains at least $\lceil N/2 \rceil$ errors when \emph{at most} $t$ errors occur in each channel and $P_2$ the probability that a coordinate of the outputs $Y$ contains at least $\lceil N/2 \rceil$ errors when \emph{exactly} $t$ errors occur in each channel. It is immediate that $P_1 \leq P_2$.  Observe that the outputs $Y$ of the $N$ channels in the case of the probability~$P_2$ can be represented as the above described variant of the multiple birthday problem as follows: an output $\by_i \in Y$ with exactly $t$ errors can be considered as a throw of $t$ balls to $n$ buckets, and there are $N$ throws in total. Therefore, $P_2 = Pr[C_t(n,N,\lceil N/2 \rceil)]$ and the claim follows.
\end{proof}

In order to obtain a lower bound on the probability $1-Pr[C_t(n,N,\lceil N/2 \rceil)]$, we require an upper bound on $Pr[C_t(n,N,\lceil N/2 \rceil)]$. In the following lemma, we present such an upper bound loosely based on a recursive idea introduced in~\cite{Suzuki_et_al} for computing the exact probability of the multiple birthday problem.
\begin{lemma} \label{Lemma_MBP_variant_estimation}
	Let $s$, $q$, $n$ and $t$ be integers satisfying $2 \leq s \leq q$ and $0 \leq t \leq n$. (i) Now the probability $Pr[C_t(n,q,s)]$ is at most
	
	{\footnotesize \[
		\frac{t^s}{n^{s-1}} \sum_{i=s}^q \left( \binom{i-1}{s-1} \left( \frac{n-t}{n} \right)^{i-s} \left(1 - Pr[C_{t-1}(n-1,i-1,s)] \right) \right) \text,
		\]}
	where $Pr[C_0(n,q,s)] = 0$ and $Pr[C_t(n,q,s)] = 0$ if $q < s$. (ii) Furthermore, we obtain that
	\[
	Pr[C_t(n,q,s)] \leq \frac{t^s}{n^{s-1}} \binom{q}{s} \text.
	\]
\end{lemma}
\begin{proof}
	(i) Observe first that we clearly have $Pr[C_0(n,q,s)] = 0$. Let then $i$ be an integer such that $s \leq i \leq q$. Denote by $C_t(n,q,s,i)$ the event that after the $i$th throw of $t$ balls there exist first time at least $s$ balls in a bucket; notice that after the $i$th throw it is possible that $s$ balls appear in multiple buckets. Using this notation, we have
	\begin{equation} \label{Eq_recursive_formula}
		Pr[C_t(n,q,s)] = \sum_{i=s}^q Pr[C_t(n,q,s,i)] \text.
	\end{equation}
	The probability $Pr[C_t(n,q,s,i)]$ can be calculated based on the following facts:
	\begin{itemize}
		\item[(i)] Let $B$ be one of the buckets that first attains $s$ balls. Clearly, the bucket $B$ can be chosen in $n$ ways.
		\item[(ii)] The $s-1$ throws placing balls into $B$ before the $i$th throw can be chosen from the previous $i-1$ throws in $\binom{i-1}{s-1}$ ways. 
		\item[(iii)] As the probability that a ball of a single throw lands into $B$ is $\binom{n-1}{t-1} / \binom{n}{t}$, the probability of the event that the $s$ selected throws put balls into $B$ is equal to
		\[
		\left( \frac{\binom{n-1}{t-1}}{\binom{n}{t}} \right)^s = \left( \frac{t}{n} \right)^s \text.
		\]
		\item[(iv)] As the probability that no ball of a single throw lands into $B$ is $\binom{n-1}{t} / \binom{n}{t}$, the probability of the event that no other throw (than the $s$ selected ones) puts a ball into $B$ is equal to
		\[
		\left( \frac{\binom{n-1}{t}}{\binom{n}{t}} \right)^{i-s} = \left( \frac{n-t}{n} \right)^{i-s} \text.
		\]
		\item[(v)] Finally, let $B'$ denote the set of $n-1$ buckets other than $B$ and $P_i$ denote the probability that no bucket in $B'$ contains at least $s$ balls after the first $i-1$ throws with the conditional assumption that the events of (iii) and (iv) occur. Observe that if a ball of a throw lands into $B$, then the throw puts $t-1$ balls into the buckets of $B'$, and otherwise $t$ balls land into $B'$.
		
	\end{itemize}
	Thus, in conclusion, we have 
	\[
	Pr[C_t(n,q,s,i)] = n \cdot \binom{i-1}{s-1} \cdot \left( \frac{t}{n} \right)^s \cdot \left( \frac{n-t}{n} \right)^{i-s} \cdot P_i \text.
	\]
	Therefore, by~\eqref{Eq_recursive_formula}, we obtain that $Pr[C_t(n,q,s)]$ is equal to 
	\[
	\frac{t^s}{n^{s-1}} \sum_{i=s}^q \left( P_i \cdot \binom{i-1}{s-1} \left( \frac{n-t}{n} \right)^{i-s} \right)  \text.
	\]
	Finally, notice that $P_i$ is equal to the probability that no bucket in $B'$ contains at least $s$ balls after $s-1$ throws with $t-1$ balls and $i-s$ throws with $t$ balls have been performed in the buckets of $B'$ (corresponding to the events of (iii) and (iv), respectively). Therefore, we obtain that that $P_i \leq 1 - Pr[C_{t-1}(n-1,i-1,s)]$ since at least $t-1$ balls are thrown $i-1$ times into the $n-1$ buckets of $B'$ (by the observation in~(v)). Hence, the claim immediately follows.
	
	(ii) For the second upper bound, we first notice that by the so called \emph{hockey stick identity} for the binomial coefficients, we have
	\[
	\sum_{i=s}^q \binom{i-1}{s-1} = \binom{q}{s} \text.
	\]
	Therefore, as $(n-t)/n \leq 1$ and the probability $P_i \leq 1$, the second claim immediately follows. 	
\end{proof}

In the following theorem, the upper bound~(ii) of the previous lemma is applied to estimate the probability in Theorem~\ref{Thm_probability_majority}. Observe that according to the second claim of the theorem, the probability that the majority algorithm outputs the transmitted word is as close to one as required when the rather weak condition $n > 2t \cdot \mathbbm{e}$, where $\mathbbm{e}$ is the Napier's constant, is satisfied and $N$ is large enough.
\begin{theorem} \label{Thm_majority_probablity_Suzuki}
	Let $C$ be an $e$-error-correcting code and $\bx$ the transmitted word of $C$. The probability that the output $\bz$ of the majority algorithm is equal to the transmitted word $\bx$ is at least
	\[
	1-\frac{t^{\lceil N/2 \rceil}}{n^{\lceil N/2 \rceil-1}}\binom{N}{\lceil N/2 \rceil} \text.
	\]
	Furthermore, if $n > 2t \cdot \mathbbm{e}$, then 
	\[
	\lim_{N \to \infty} \left( 1-Pr[C_t(n,N,\lceil N/2 \rceil)] \right) = 1 \text.
	\]
\end{theorem} 
\begin{proof}
	The first claim follows immediately by applying Theorem~\ref{Thm_probability_majority} and Lemma~\ref{Lemma_MBP_variant_estimation}. Assuming $m$ and $r$ are integers such that $1 \leq r \leq m$, we have the following well-known upper bound on the binomial coefficient:
	\[
	\binom{m}{r} < \left( \frac{m \cdot \mathbbm{e}}{r} \right)^r \text. 
	\]
	Applying this upper bound, we obtain that
	\[
	\begin{split}
		Pr[C_t(n,N,\lceil N/2 \rceil)]  &\leq \frac{t^{\lceil N/2 \rceil}}{n^{\lceil N/2 \rceil-1}}\binom{N}{\lceil N/2 \rceil} \\ &<  \frac{t^{\lceil N/2 \rceil}}{n^{\lceil N/2 \rceil-1}}\left( \frac{N \cdot \mathbbm{e}}{\lceil N/2 \rceil} \right)^{\lceil N/2 \rceil} \\ & \leq 2t \cdot \mathbbm{e} \left( \frac{2t \cdot \mathbbm{e}}{n} \right)^{\lceil N/2 \rceil - 1} \to 0 \text.
	\end{split}
	\]
	as $N \to \infty$ since $n > 2t \cdot \mathbbm{e}$. Thus, the second claim follows.
\end{proof}

\begin{table}
	\centering
	\caption{The lower bound of Theorem~\ref{Thm_probability_majority} together with Lemma~\ref{Lemma_MBP_variant_estimation} and Theorem~\ref{Thm_majority_probablity_Suzuki} as well as the Monte Carlo approximations with $100000$ samples of the probability that $\bz = \bx$ when $n=28$, $t=5$ and $N=11,21,31,41,101$.} \label{Table_probabilities}
	\begin{tabular}{|c|c|c|c|}
		\hline
		$N$ & Theorem~\ref{Thm_probability_majority} and Lemma~\ref{Lemma_MBP_variant_estimation} & Theorem~\ref{Thm_majority_probablity_Suzuki} & Simulation \\ \hline
		11  &                                      $0.8199$                                       &                   $0.5806$                   &  $0.8615$  \\ \hline
		21  &                                      $0.9901$                                       &                   $0.9419$                   &  $0.9929$  \\ \hline
		31  &                                      $0.9994$                                       &                   $0.9910$                   &  $0.9997$  \\ \hline
		41  &                                      $0.9997$                                       &                   $0.9985$                   &  $1.000$   \\ \hline
		101 &                                       $1.000$                                       &                   $1.000$                    &  $1.000$   \\ \hline
	\end{tabular}    
\end{table}
In Table~\ref{Table_probabilities}, we illustrate the various approaches to approximate the probability that the output~$\bz$ of the majority algorithm is equal to the transmitted word $\bx$ when $t=5$ and $n=28$ (satisfying the condition of Theorem~\ref{Thm_majority_probablity_Suzuki}): the lower bound on the probabilities of Theorem~\ref{Thm_probability_majority} together with Lemma~\ref{Lemma_MBP_variant_estimation} and Theorem~\ref{Thm_majority_probablity_Suzuki} as well as the Monte Carlo simulations with $100000$ samples. Observe that the majority algorithm outputs $\bz = \bx$ with high probability for significantly smaller number $N$ of channels compared to Theorem~\ref{L=1}, for which the required number of channels is $41709$ when $e=0$, $t=5$ and $n=28$; here $e=0$ is chosen in order to meet the requirement $\bz = \bx$.

Above, we saw that it is highly probable that the majority algorithm works correctly with significantly smaller number of channels compared to the one given in Theorem~\ref{L=1}, which was required in the algorithm presented in~\cite{Uusi_Maria_Abu-Sini}. In what follows, we take another approach on the majority algorithm and show that if a certain criteria (see Theorem~\ref{Thm_majority_algorithm}) is met for the outputs $Y$, then we can verify that the output $\bz$ of the majority algorithm belongs to $B_e(\bx)$. For this purpose, notice first that the total number of errors occurring in the $i$th coordinates of the outputs $Y$ is at least $m_i = \min\{m_{i,0},m_{i,1}\}$. On the other hand, there happens at most $t$ errors in each channel and, hence, the total number of errors in the channels is at most $tN$. Thus, we obtain that
\begin{equation} \label{Eq_error_sum}
	\sum_{i=1}^n m_i \leq tN \text.
\end{equation}
Furthermore, if $\bx = \bz$, then the number of errors is exactly $\sum_{i=1}^n m_i$. In addition, if $\bx \neq \bz$, then for each coordinate $i$ in which the words differ, $\max\{m_{i,0},m_{i,1}\} = N-m_i$ is contributed to the sum of errors (instead of $m_i$). The following theorem is based on the idea that even the modified sum (in the left hand side of \eqref{Eq_sum_condition}) has to satisfy Inequality~\eqref{Eq_error_sum}.
\begin{theorem} \label{Thm_majority_algorithm}
	Let $C$ be an $e$-error-correcting code, $m'_i$ be the integers $m_i$ ordered in such a way that $m'_1 \geq m'_2 \geq \cdots \geq m'_n$ and  $\bz$ be the output word of the majority algorithm. We have $d(\bx,\bz)\leq k$ if $k$ is a positive integer such that
	\begin{equation} \label{Eq_sum_condition}
		\sum_{i=1}^{k+1} (N-m'_i) + \sum_{i=k+2}^n m'_i > tN \text.
	\end{equation}
\end{theorem}
\begin{proof}
	Let $k$ be a positive integer satisfying \eqref{Eq_sum_condition}. Suppose to the contrary that $d(\bx,\bz) \geq k+1$. This implies that for at least $k+1$ coordinates $i$, the number of errors is $\max\{m_{i,0},m_{i,1}\} = N-m_i$. Therefore, by the ordering of $m'_i$, the number of errors is at least 
	\[
	\sum_{i=1}^{k+1} (N-m'_i) + \sum_{i=k+2}^n m'_i \ (> tN) \text.
	\]
	Thus, due to \eqref{Eq_sum_condition}, we have a contradiction with the maximum number of errors being $tN$ and the claim follows.
\end{proof}

Observe that~\eqref{Eq_sum_condition} allows us to estimate the accuracy of $\bz$. In particular, if $k \leq e$, then $d(\bx, \bz) \leq k \leq e$ and the word $\bz$ can be decoded to $\bx$ as $C$ is an $e$-error-correcting code. Furthermore, if $k > e$, then $\bx \in C \cap B_k(\bz)$ and the decoding algorithm outputs a list of words containing $\bx$. Moreover, the size of the list is at most $\max_{\bu \in \F^n} |C \cap B_k(\bu)|$, which is closely related to the traditional list decoding (see~\cite{Guruswamin_kirja}). In conclusion, the theorem gives us a condition guaranteeing that the transmitted word can be decoded uniquely or with certain accuracy.

\begin{table}
	\centering
	\caption{The Monte Carlo approximations with $100000$ samples of the probability for $e$ satisfying the condition of Theorem~\ref{Thm_majority_algorithm} when  $n=24$, $t=7$, $e=2,3,4$ and $N=11,21,31,41$.} \label{Table_verified_probabilities}
	\begin{tabular}{|c|c|c|c|}
		\hline
		$N \backslash e$ &   $2$   &   $3$   &   $4$   \\ \hline
		11        & $0.068$ & $0.260$ & $0.587$ \\ \hline
		21        & $0.369$ & $0.790$ & $0.972$ \\ \hline
		31        & $0.701$ & $0.971$ & $0.999$ \\ \hline
		41        & $0.887$ & $0.997$ & $0.999$ \\ \hline
	\end{tabular}    
\end{table}

The probability, that in Theorem~\ref{Thm_majority_algorithm} there exists a certain $k$ satisfying~\eqref{Eq_sum_condition}, can be analysed analytically as shown below, but first we approximate it using Monte Carlo simulations. In Table~\ref{Table_verified_probabilities}, the probability is approximated using $100000$ samples for $n=24$, $t=7$, $e=2,3,4$ and varying number of channels $N$; here we choose $k = e$ and we strive for the exact transmitted word $\bx$. From the table, we can notice that as the number of channels increases it becomes very likely that Inequality~\eqref{Eq_sum_condition} is satisfied and that the majority algorithm together with the $e$-error-correction capability of $C$ verifiably outputs the transmitted word $\bx$.

In what follows, we further study analytically the probability that for a set $Y$ of outputs there exists an integer $k$ satisfying the conditions of Theorem~\ref{Thm_majority_algorithm}. For this purpose, let $C_t'(n,q,s)$ denote the event that after $q$ random throws of $t$ balls, each of the $n$ buckets contains at most $s$ balls. Furthermore, regarding the number of errors occurring in the channels (in total), let $Er(r)$ denote the event that at least $r$ errors happen in the outputs $Y$ in total. Moreover, let $p(N)$ denote the parity of $N$, i.e., $p(N) = 1$ if $N$ is odd, and otherwise $p(N) = 0$. Now we are ready to formulate the following theorem.
\begin{theorem} \label{Thm_majority_verifiable_probability}
	Let $C$ be an $e$-error-correcting code and $\alpha$ be a positive integer smaller than $\lceil N/2 \rceil$. The probability that a positive integer $k$ satisfies~\eqref{Eq_sum_condition} is at least
	\[
	Pr[C_t'(n,N,\lceil N/2 \rceil-\alpha) \cap Er(tN-(k+1)(2\alpha-p(N))+1)] \text. 
	\]
\end{theorem}
\begin{proof}
	Assume that (i) at most $\lceil N/2 \rceil-\alpha$ errors occur in each coordinate of the outputs $Y$ and that (ii) at least $tN-(k+1)(2\alpha-p(N))+1$ errors occur in the channels in total (when at most $t$ errors occur in each channel). Now the difference $(N-m_i)-m_i = N-2m_i$ gives the number of additional errors occurring in a coordinate in the case that $\max\{m_{i,0},m_{i,1}\} = N-m_i$ errors happen instead of $m_i = \min\{m_{i,0},m_{i,1}\}$ errors. The difference $(N-m_i)-m_i$ can be estimated based on the parity of $N$ as follows:
	\begin{itemize}
		\item If $N$ is even, then $N-2m_i \geq N-2(\lceil N/2 \rceil-\alpha) = 2\alpha = 2\alpha-p(N)$.
		\item If $N$ is odd, then $N-2m_i \geq N-2(\lceil N/2 \rceil-\alpha) = 2\alpha-1 = 2\alpha-p(N)$.
	\end{itemize}
	In conclusion, we have $N-2m_i \geq 2\alpha-p(N)$. Therefore, using the notation of Theorem~\ref{Thm_majority_algorithm}, we have 
	\[
	\begin{split}
		&\quad \sum_{i=1}^{k+1} (N-m'_i) + \sum_{i=k+2}^n m'_i \\ &\geq \sum_{i=1}^n m'_i + (k+1)(2\alpha-p(N)) > tN  
	\end{split}
	\]
	since $\sum_{i=1}^nm'_i\geq tN-(k+1)(2\alpha-2p(N))+1$ by the assumption~(ii). Thus, the integer $k$ satisfies the condition of Theorem~\ref{Thm_majority_algorithm}.
	
	In order to estimate the probability of the events of the assumptions~(i) and (ii) occurring simultaneously, we denote by $A$ the event that at most $\lceil N/2 \rceil-\alpha$ errors occur in each coordinate of the outputs $Y$ when at most $t$ errors happen in each channel. Clearly, since the event $C_t'(n,N,\lceil N/2 \rceil-\alpha)$ can be interpreted as the outputs $Y$ of the channels as in Theorem~\ref{Thm_probability_majority}, we have $C_t'(n,N,\lceil N/2 \rceil-\alpha) \subseteq A$. Hence, the probability of the events of the assumptions~(i) and (ii) occurring simultaneously is at least $Pr[A \cap Er(tN-(k+1)(2\alpha-p(N))+1)] \geq Pr[C_t'(n,N,\lceil N/2 \rceil-\alpha) \cap Er(tN-(k+1)(2\alpha-p(N))+1)]$ and the claim follows.
\end{proof}

Observe that if $A$ and $B$ are events (independent or dependent), then the probability of the event $A \cap B$ can be estimated as follows by the inclusion-exclusion principle:
\[
Pr[A \cap B] = Pr[A] + Pr[B] - Pr[A \cup B] \geq Pr[A] + Pr[B] - 1 \text.
\]
Applying this lower bound to Theorem~\ref{Thm_majority_verifiable_probability}, the following corollary is immediately obtained.
\begin{corollary} \label{Cor_lower_bound_with_alpha}
	Let $C$ be an $e$-error-correcting code and $\alpha$ be a positive integer. The probability that a positive integer $k$ satisfies~\eqref{Eq_sum_condition} is at least $Pr[C_t'(n,N,\lceil N/2 \rceil-\alpha)] + Pr[Er(tN-(k+1)(2\alpha-p(N))+1)] - 1$.
\end{corollary}

In what follows, we discuss how the value $\alpha$ in the previous corollary should be chosen in order for the lower bound to be as close to $1$ as desired. Consider first the approximation of the probability $Pr[Er(tN-(k+1)(2\alpha-p(N))+1)]$. Let $X$ be a categorical random variable representing the number of errors occurring in a channel. Hence, for $r = 0, 1, \ldots, t$, the probability
\[
Pr[X=r] = p_r = \frac{\binom{n}{r}}{\sum_{i=0}^{t} \binom{n}{i}} \text.
\]
The expected value of the distribution corresponding to $X$ is
\[
\mu = E(X) = \sum_{r=0}^{t} rp_r 
\]
and the variance is
\[
\sigma^2 = Var(X) = E(X^2) - \mu^2 = \sum_{r=0}^{t} r^2p_r - \mu^2 \text.
\]
Let
\[
S_N = X_1 + X_2 + \cdots + X_N 
\]
be the sum of independent and identically distributed random variables $X_i$ each equal to $X$. Clearly, we have $Pr[Er(tN-(k+1)(2\alpha-p(N))+1)] = Pr[S_N \geq tN-(k+1)(2\alpha-p(N))+1]$.
By the Central Limit Theorem (for example, see~\cite[Section~5.4]{Tijms_Understanding_Probability}), the sum $S_N$ can be approximated (for large enough values of $N$) by the normal distribution $\mathcal{N}(\mu N,\sigma\sqrt{N})$, where $\mu N$ and $\sigma\sqrt{N}$ are the expected value and the standard deviation of $S_N$, respectively. 
A usual approach on the normal distribution is to use the following concept of \emph{confidence intervals}. For a positive integer $h$, the probability
\[
Pr[\mu N - h \cdot \sigma\sqrt{N} \leq S_N \leq \mu N + h \cdot \sigma\sqrt{N}]
\]
rapidly approaches one as $h$ increases; even for $h = 4$ the probability is approximately $0.999937$. Straightforwardly, we obtain that $tN-(k+1)(2\alpha-p(N))+1 \leq \mu N - h \cdot \sigma\sqrt{N}$ if and only if
\[
\alpha \geq \frac{(t-\mu)N+h \cdot \sigma\sqrt{N} + 1 + p(N)(k+1)}{2(k+1)} (> 0) \text.
\]
Now we have $Pr[S_N \geq tN-(k+1)(2\alpha-p(N))+1] \geq Pr[\mu N - h \cdot \sigma\sqrt{N} \leq S_N \leq \mu N + h \cdot \sigma\sqrt{N}]$.

In order to estimate the lower bound of Corollary~\ref{Cor_lower_bound_with_alpha} for some fixed $n$, $e$ and $t$, we first choose a suitable positive integer $h$ such that the previous lower bound of $Pr[S_N \geq tN-(k+1)(2\alpha-p(N))+1]$ is as large as desired. Then, according to the following lemma, we have a lower bound on the probability $Pr[C'(n,tN,\lceil N/2 \rceil - \alpha)]$, which approaches $1$ as the number $N$ of channels tends to infinity. Hence, for carefully chosen $h$, $\alpha$ and $N$, the probability discussed in Corollary~\ref{Cor_lower_bound_with_alpha} gets as close to $1$ as desired. Furthermore, notice that the proof of the lemma is based on similar ideas as the one of Theorem~\ref{Thm_majority_probablity_Suzuki}.
\begin{lemma} \label{Lemma_lower_bound_with_alpha}
	Let $h$ and $k$ be (fixed) positive integers and $\alpha = \lceil ((t-\mu)N+h \cdot \sigma\sqrt{N} + 1 + p(N)(k+1)) / (2(k+1)) \rceil$. Then the probability $Pr[C_t'(n,N,\lceil N/2 \rceil - \alpha)]$ is at least
	\[
	1-\frac{t^{\lceil N/2 \rceil - \alpha + 1}}{n^{\lceil N/2 \rceil - \alpha}}\binom{N}{\lceil N/2 \rceil  - \alpha + 1} \text.
	\]
	Furthermore, denoting $A = (k-(t-\mu))/(k+1)$, we obtain that if $t-\mu < k$ and $n >2t \cdot \mathbbm{e}/A$, then 
	\[
	\lim_{N \to \infty} Pr[C_t'(n,N,\lceil N/2 \rceil - \alpha)]  = 1 \text.
	\]
\end{lemma}
\begin{proof}
	The lower bound on $Pr[C_t'(n,N,\lceil N/2 \rceil - \alpha)] = 1- Pr[C_t(n,N,\lceil N/2 \rceil - \alpha + 1)]$ immediately follows by Lemma~\ref{Lemma_MBP_variant_estimation}(ii) and hence, the first claim follows. Assume then that $t-\mu < k$ and $n >2t \cdot \mathbbm{e}/A$. For the limit of second claim, observe first that $h \cdot \sigma\sqrt{N} + 1 + p(N)(k+1) \leq N$ when $N$ is large enough since $k$ is a (fixed) constant. Therefore, by straightforward calculations, we have
	\begin{equation} \label{Eq_alpha_upper_bound}
		\alpha \leq \frac{(t-\mu+1)N}{2(k+1)} + 1 \text.
	\end{equation}
	This further implies that 
	\[
	1 - \frac{2(\alpha-1)}{N} \geq 1 - \frac{t-\mu+1}{k+1} = \frac{k-(t-\mu)}{k+1} = A\text.
	\]
	Applying this inequality, for large enough $N$, we have 
	\[
	\begin{split}
		&\quad Pr[C_t'(n,N,\lceil N/2 \rceil - \alpha)] \\ &> 1 - \frac{t^{\lceil N/2 \rceil - \alpha + 1}}{n^{\lceil N/2 \rceil - \alpha}} \left( \frac{N \cdot \mathbbm{e}}{N/2 - \alpha + 1} \right)^{\lceil N/2 \rceil - \alpha + 1} \\ &= 1 -  2t  \mathbbm{e} \left( \frac{2t \mathbbm{e}}{n}\right)^{\lceil N/2 \rceil - \alpha}  \left( \frac{1}{1-\frac{2(\alpha-1)}{N}}\right)^{\lceil N/2 \rceil - \alpha + 1} \\ &\geq 1 -  2t  \mathbbm{e} \left( \frac{2t \mathbbm{e}}{n}\right)^{\lceil N/2 \rceil - \alpha}  \left( \frac{1}{A}\right)^{\lceil N/2 \rceil - \alpha + 1} \\ &= 1 -  \frac{2t  \mathbbm{e}}{A} \left( \frac{2t \mathbbm{e}}{An}\right)^{\lceil N/2 \rceil - \alpha}  \text.
	\end{split}
	\]
	Now the probability
	\[
	Pr[C_t'(n,N,\lceil N/2 \rceil - \alpha)] > 1 - \frac{2t\mathbbm{e}}{A} \left( \frac{2t \mathbbm{e}}{An}\right)^{\lceil N/2 \rceil - \alpha}
	\]
	and the lower bound approaches one as $N$ tends to infinity since $n >2t \cdot \mathbbm{e}/A$ and $\lceil N/2 \rceil - \alpha \to \infty$ due to Inequality~\eqref{Eq_alpha_upper_bound} and the assumption $t-\mu < k$. Thus, the second claim follows.
\end{proof}

Notice that the condition $t-\mu < k$ of the lemma is rather undemanding for our purposes. Observe that the lemma is usually applied for $k=e$ as then the output $\bz$ of the majority algorithm can be uniquely decoded to the transmitted word $\bx$ according to Theorem~\ref{Thm_majority_algorithm}. In this case, the condition can be formulated as $\ell<\mu$. Moreover, we trivially have $\mu>tp_t= t\binom{n}{t} / V(n,t)$. Furthermore, we have the estimation 
\[
\binom{n}{t-a}\leq\frac{t^a}{\prod_{i=1}^a n-t+i}\binom{n}{t}\leq \frac{t^a}{(n-t+1)^a}\binom{n}{t}.
\]
Thus, if $n\geq t^2+t-1$, we obtain
\[
\mu>\frac{t}{\sum_{i=0}^t\frac{t^i}{(n-t+1)^i}} \geq \frac{t}{\sum_{i=0}^t\frac{1}{t^i}}>\frac{t}{\sum_{i=0}^{\infty}\frac{1}{t^i}}=t-1 \text,
\]
where the assumption $n\geq t^2+t-1$ is used in the second inequality. Furthermore, we have $t-1 \geq \ell$ when $e \geq 1$. Therefore, as $\mu>t-1$, it is actually enough for the requirement $\ell<\mu$ that $t-1\geq\ell$. In other words, condition $t-\mu < e$ is satisfied when $n\geq t^2+t-1$ and $e\geq 1$. Moreover, even this quite meagre condition for  $n$ is more than we actually need due to rather crude estimates;  for example, the condition $t-\mu < e$ is met for the values $n=10$, $e=1$ and $t=5$ (as well as for all $n > 10$ for the same choices of $e$ and $t$).

As explained before the previous lemma, we can get the lower bound of Corollary~\ref{Cor_lower_bound_with_alpha} as close to $1$ as desired for carefully chosen $h$, $\alpha$ and $N$ by combining Lemma~\ref{Lemma_lower_bound_with_alpha} with the discussions on the probability $Pr[Er(tN-(k+1)(2\alpha-p(N))+1)] = Pr[S_N \geq tN-(k+1)(2\alpha-p(N))+1]$. For example, if $n=62$, $t=7$, $(k=) \ e=2$ and $h=3$, then the conditions of Lemma~\ref{Lemma_lower_bound_with_alpha} are satisfied and the lower bound of Corollary~\ref{Cor_lower_bound_with_alpha} is approximately $0.997$ when $N = 41$.
%

\bibliographystyle{IEEEtran}
\bibliography{ISIT2019}

\begin{thebibliography}{10}
\providecommand{\url}[1]{#1}
\csname url@samestyle\endcsname
\providecommand{\newblock}{\relax}
\providecommand{\bibinfo}[2]{#2}
\providecommand{\BIBentrySTDinterwordspacing}{\spaceskip=0pt\relax}
\providecommand{\BIBentryALTinterwordstretchfactor}{4}
\providecommand{\BIBentryALTinterwordspacing}{\spaceskip=\fontdimen2\font plus
\BIBentryALTinterwordstretchfactor\fontdimen3\font minus
  \fontdimen4\font\relax}
\providecommand{\BIBforeignlanguage}[2]{{%
\expandafter\ifx\csname l@#1\endcsname\relax
\typeout{** WARNING: IEEEtran.bst: No hyphenation pattern has been}%
\typeout{** loaded for the language `#1'. Using the pattern for}%
\typeout{** the default language instead.}%
\else
\language=\csname l@#1\endcsname
\fi
#2}}
\providecommand{\BIBdecl}{\relax}
\BIBdecl

\bibitem{junnila2022list}
V.~Junnila, T.~Laihonen, and T.~Lehtil{\"a}, ``On the list size in the
  {L}evenshtein's sequence reconstruction problem,'' in \emph{2022 IEEE
  International Symposium on Information Theory (ISIT)}.\hskip 1em plus 0.5em
  minus 0.4em\relax IEEE, 2022, pp. 510--515.

\bibitem{Levenshtein}
V.~I. Levenshtein, ``Efficient reconstruction of sequences,'' \emph{IEEE Trans.
  Inform. Theory}, vol.~47, no.~1, pp. 2--22, 2001.

\bibitem{levenshtein2005reconstruction}
V.~Levenshtein, E.~Konstantinova, E.~Konstantinov, and S.~Molodtsov,
  ``Reconstruction of a graph from 2-vicinities of its vertices,''
  \emph{Discrete Appl. Math.}, vol. 156, pp. 1399--1406, 2008.

\bibitem{gabrys2018sequence}
R.~Gabrys and E.~Yaakobi, ``Sequence reconstruction over the deletion
  channel,'' \emph{IEEE Trans. Inform. Theory}, vol.~64, no.~4, pp. 2924--2931,
  2018.

\bibitem{horovitz2018reconstruction}
M.~Horovitz and E.~Yaakobi, ``Reconstruction of sequences over non-identical
  channels,'' \emph{IEEE Trans. Inform. Theory}, vol.~65, no.~2, pp.
  1267--1286, 2018.

\bibitem{Maria_Abu-Sini}
M.~Abu-Sini and E.~Yaakobi, ``On list decoding of insertions and deletions
  under the reconstruction model,'' in \emph{Proceedings of 2021 IEEE
  International Symposium on Information Theory}, 2021, pp. 1706--1711.

\bibitem{Uusi_Maria_Abu-Sini}
------, ``On {L}evenshtein's reconstruction problem under insertions,
  deletions, and substitutions,'' \emph{IEEE Trans. Inform. Theory}, vol.~67,
  no.~11, pp. 7132--7158, 2021.

\bibitem{levenshtein2001efficient}
V.~I. Levenshtein, ``Efficient reconstruction of sequences from their
  subsequences or supersequences,'' \emph{Journal of Combinatorial Theory,
  Series A}, vol.~93, no.~2, pp. 310--332, 2001.

\bibitem{yaakobi2016constructions}
E.~Yaakobi, J.~Bruck, and P.~H. Siegel, ``Constructions and decoding of cyclic
  codes over $ b $-symbol read channels,'' \emph{IEEE Trans. Inform. Theory},
  vol.~62, no.~4, pp. 1541--1551, 2016.

\bibitem{bornholt2016dna}
J.~Bornholt, R.~Lopez, D.~M. Carmean, L.~Ceze, G.~Seelig, and K.~Strauss, ``A
  {DNA}-based archival storage system,'' \emph{ACM SIGARCH Computer
  Architecture News}, vol.~44, no.~2, pp. 637--649, 2016.

\bibitem{church2012next}
G.~M. Church, Y.~Gao, and S.~Kosuri, ``Next-generation digital information
  storage in {DNA},'' \emph{Science}, vol. 337, no. 6102, pp. 1628--1628, 2012.

\bibitem{grass2015robust}
R.~N. Grass, R.~Heckel, M.~Puddu, D.~Paunescu, and W.~J. Stark, ``Robust
  chemical preservation of digital information on {DNA} in silica with
  error-correcting codes,'' \emph{Angew. Chem. Int. Edit.}, vol.~54, no.~8, pp.
  2552--2555, 2015.

\bibitem{yazdi2015dna}
S.~H.~T. Yazdi, H.~M. Kiah, E.~Garcia-Ruiz, J.~Ma, H.~Zhao, and O.~Milenkovic,
  ``{DNA}-based storage: Trends and methods,'' \emph{IEEE Transactions on
  Molecular, Biological and Multi-Scale Communications}, vol.~1, no.~3, pp.
  230--248, 2015.

\bibitem{yaakobi2018uncertainty}
E.~Yaakobi and J.~Bruck, ``On the uncertainty of information retrieval in
  associative memories,'' \emph{IEEE Trans. Inform. Theory}, vol.~65, no.~4,
  pp. 2155--2165, 2018.

\bibitem{YBiram}
------, ``On the uncertainty of information retrieval in associative
  memories,'' in \emph{Proceedings of 2012 IEEE International Symposium on
  Information Theory}, 2012, pp. 106--110.

\bibitem{cheraghchi2020coded}
M.~Cheraghchi, R.~Gabrys, O.~Milenkovic, and J.~Ribeiro, ``Coded trace
  reconstruction,'' \emph{IEEE Transactions on Information Theory}, vol.~66,
  no.~10, pp. 6084--6103, 2020.

\bibitem{JLirvnic}
V.~Junnila and T.~Laihonen, ``Information retrieval with varying number of
  input clues,'' \emph{IEEE Trans. Inform. Theory}, vol.~62, no.~2, pp.
  625--638, 2016.

\bibitem{JLcirsu}
------, ``Codes for information retrieval with small uncertainty,'' \emph{IEEE
  Trans. Inform. Theory}, vol.~60, no.~2, pp. 976--985, 2014.

\bibitem{laihonen2017improved}
T.~Laihonen and T.~Lehtil{\"a}, ``Improved codes for list decoding in the
  {L}evenshtein's channel and information retrieval,'' in \emph{Proceedings of
  2017 IEEE International Symposium on Information Theory}, 2017, pp.
  2643--2647.

\bibitem{laihonen2019improved}
T.~Laihonen, ``On t-revealing codes in binary {H}amming spaces,''
  \emph{Information and Computation}, vol. 268, 2019.

\bibitem{batu2004reconstructing}
T.~Batu, S.~Kannan, S.~Khanna, and A.~McGregor, ``Reconstructing strings from
  random traces,'' in \emph{Proceedings of the fifteenth annual ACM-SIAM
  symposium on Discrete algorithms}, 2004, pp. 910--918.

\bibitem{viswanathan2008improved}
K.~Viswanathan and R.~Swaminathan, ``Improved string reconstruction over
  insertion-deletion channels,'' in \emph{Proceedings of the nineteenth annual
  ACM-SIAM symposium on Discrete algorithms}, 2008, pp. 399--408.

\bibitem{junnila2020levenshtein}
V.~Junnila, T.~Laihonen, and T.~Lehtil{\"a}, ``On {L}evenshtein's channel and
  list size in information retrieval,'' \emph{IEEE Trans. Inform. Theory},
  vol.~67, no.~6, pp. 3322--3341, 2020.

\bibitem{chll}
G.~Cohen, I.~Honkala, S.~Litsyn, and A.~Lobstein, \emph{Covering codes}, ser.
  North-Holland Mathematical Library.\hskip 1em plus 0.5em minus 0.4em\relax
  Amsterdam: North-Holland Publishing Co., 1997, vol.~54.

\bibitem{sauer1972density}
N.~Sauer, ``On the density of families of sets,'' \emph{J. Comb. Theory A},
  vol.~13, no.~1, pp. 145--147, 1972.

\bibitem{shelah1972combinatorial}
S.~Shelah, ``A combinatorial problem; stability and order for models and
  theories in infinitary languages,'' \emph{Pac. J. Math.}, vol.~41, no.~1, pp.
  247--261, 1972.

\bibitem{Guruswamin_kirja}
\BIBentryALTinterwordspacing
V.~Guruswami, \emph{List decoding of error-correcting codes}.\hskip 1em plus
  0.5em minus 0.4em\relax ProQuest LLC, Ann Arbor, MI, 2001, thesis
  (Ph.D.)--Massachusetts Institute of Technology. [Online]. Available:
  \url{http://gateway.proquest.com/openurl?url_ver=Z39.88-2004&rft_val_fmt=info:ofi/fmt:kev:mtx:dissertation&res_dat=xri:pqdiss&rft_dat=xri:pqdiss:0803408}
\BIBentrySTDinterwordspacing

\bibitem{Suzuki_et_al}
K.~Suzuki, D.~Tonien, K.~Kurosawa, and K.~Toyota, ``Birthday paradox for
  multi-collisions,'' in \emph{Information Security and Cryptology -- ICISC
  2006}, M.~S. Rhee and B.~Lee, Eds.\hskip 1em plus 0.5em minus 0.4em\relax
  Berlin, Heidelberg: Springer Berlin Heidelberg, 2006, pp. 29--40.

\bibitem{Kounavis_et_al}
M.~Kounavis, S.~Deutsch, D.~Durham, and S.~Komijani, ``Non-recursive
  computation of the probability of more than two people having the same
  birthday,'' in \emph{2017 IEEE Symposium on Computers and Communications
  (ISCC)}, 2017, pp. 1263--1270.

\bibitem{Tijms_Understanding_Probability}
H.~Tijms, \emph{\BIBforeignlanguage{eng}{Understanding probability}},
  3rd~ed.\hskip 1em plus 0.5em minus 0.4em\relax Cambridge: Cambridge
  University Press, 2012.

\end{thebibliography}


\section*{Appendix}\label{appendix}
\begin{proof}[Proof of Corollary \ref{YB raja e-error l=2COR}.]
Let us now consider the proof of Corollary \ref{YB raja e-error l=2COR} in the case where $d=2e+1$. The proof is quite similar to the case with $d=2e+2$. Although the first half of the proof in Corollary \ref{YB raja e-error l=2COR} works for $d=2e+1$, we start from the beginning since using set  $W'_w$ instead of $W_w$ is better suited for our goal.

 First we study the second binomial sum in the claim of Theorem \ref{KanavalukemaEkviv} when $h=3$. Then we show that it gives equivalent result with Theorem \ref{YB raja e-error l=2COR} when $d=2e+1$. Since the binomial sum in Equation (\ref{eqSa'}) is equivalent with the binomial sum in Equation (\ref{eqSa}), the claim follows.

Let us now consider the value we get for $N_3$ in the second binomial sum of Theorem \ref{KanavalukemaEkviv} when $h=3$. We have

\begin{align*}
N_3=&V(n,\ell-1)+\sum_{w\geq\ell}\sum_{(i_2,i_3,i_4)\in W'_w}\binom{n-3e-2}{w-i_2-i_3-i_4}
\binom{e+1}{i_2}\binom{e+1}{i_3}\binom{e}{i_4}.
\end{align*}
In what follows, we have renamed the indices for convenience in such a way that $i_4$ corresponds to $i_1$ of Theorem~\ref{KanavalukemaEkviv} (the index $i_1$ will be saved for later use in the proof). Moreover, we have $W'_w=\{(i_2,i_3,i_4)\mid \text{for } j\in[2,3]:  (w+1-\ell)/2\leq  i_j\leq e+1 \text{ and } (w-\ell)/2\leq  i_4\leq e, w\geq i_2+i_3+i_4\}$. 
Again, if we have some binomial coefficients with $i_j<0$ for some $j$, then we use the common convention that the binomial coefficient attains the value $0$. Then, as in the case $d=2e+2$, we obtain that $N_3=\sum_{w\geq0}\sum_{(i_2,i_3,i_4)\in W'_w}\binom{n-3e-2}{w-i_2-i_3-i_4}\binom{e+1}{i_2}\binom{e+1}{i_3}\binom{e}{i_4}$.


Let us denote by $i_1=w-i_2-i_3-i_4$. We again get that $2i_2\geq i_1+i_2+i_3+i_4-\ell+1$ and similar inequality for $i_3$. Moreover, for $i_4$, we have $2i_4\geq i_1+i_2+i_3+i_4-\ell$. Since we do not have to take into account lower bound $i_1\geq0$ (cases with $i_1<0$ increase binomial sum by $0$) or the cases with $i_j>e+1$ for $j\in\{2,3,4\}$, we can consider following system of inequalities:
\begin{align}
 i_2\geq& i_1+i_3+i_4-\ell+1\label{Ai_2}\\
i_3\geq& i_1+i_2+i_4-\ell+1\label{Ai_3}\\
 i_4\geq& i_1+i_2+i_3-\ell\label{Ai_4}. 
 \end{align}

Our goal is to show that this system of inequalities is equivalent with the following system of inequalities: 
\begin{align}
i_4&\leq \ell-1-i_1, \label{Ai'_4}\\
i_3&\leq \ell-1-i_1,  \label{Ai'_3}\\
i_1+i_3+i_4-(\ell-1)\leq i_2&\leq \ell-1/2-i_1-|i_4-i_3+1/2|.\label{Ai'_2}
\end{align}

Let us first show that the second system of inequalities follows from the first system of inequalities.

Inequality (\ref{Ai'_4}) follows from $$i_4=w-i_1-i_2-i_3\leq w-i_1-2(w-\ell+1)/2=\ell-1-i_1.$$ We obtain Inequality (\ref{Ai'_3}) in similar manner. Indeed, $$i_3=w-i_1-i_2-i_4\leq w-i_1-(w-\ell+1)/2-(w-\ell)/2=\ell-1/2-i_1$$ and the upper bound follows from the fact that $i_3$ is an integer. Moreover, from Inequalities (\ref{Ai_3}) and (\ref{Ai_4}) we obtain $i_2\leq \ell-1/2-i_1-i_4-1/2+i_3$ and $i_2\leq \ell-1/2-i_1-i_3+i_4+1/2$, respectively. Together, these imply $$i_2\leq \ell-1/2-i_1-|i_4+1/2-i_3|.$$ Finally, the lower bound inequality in (\ref{Ai'_2}) follows directly from (\ref{Ai_2}).

Let us then show that the first system of inequalities follows from the second one. First of all, Inequality (\ref{Ai_2}) follows immediately from Inequality (\ref{Ai'_2}). Assume first that $i_4+1/2> i_3$. Then the upper bound of Inequality (\ref{Ai'_2}) is $i_2\leq \ell-1-i_1-i_4+i_3$. This implies Inequality (\ref{Ai_3})  and inequality (\ref{Ai_4}) since $$i_4\geq i_3\geq i_1+i_2+i_4-\ell+1\geq i_1+i_2+i_3-\ell+1.$$ Notice that we cannot attain the lower bound in Inequality (\ref{Ai_4}) in this case.

When $i_3> i_4+1/2$, then the upper bound of Inequality (\ref{Ai'_2}) is $i_2\leq \ell-i_1-i_3+i_4$. In this case $i_4\geq i_1+i_2+i_3-\ell$ and since $i_3>i_4$, both lower bounds,   (\ref{Ai_3}) and (\ref{Ai_4}), follow.

Finally, we may add lower bounds $i_j\geq0$ for all $j\in\{1,2,3,4\}$ due to binomial coefficient context. Similarly we notice that if $i_1\geq\ell$, then $i_3<0$. Thus, we may also add upper bound  $i_1\leq \ell-1$.  Now, we are ready to compare this bound with the bound of Theorem \ref{YB raja l=2} by Yaakobi and Bruck.

When we have $d=2e+1$, Theorem \ref{YB raja l=2} can be presented in the following way: Let 

\noindent$N\geq\sum_{h_1,h_2,h_3,h_4}\binom{n-3e-2}{h_1}\binom{e+1}{h_2}\binom{e+1}{h_3}\binom{e}{h_4}+1$ for \begin{itemize}
\item $0\leq h_1\leq \ell-1$,
\item $h_1-\ell\leq h_4\leq \ell-1-h_1$, 
\item $e+2-\ell+h_1\leq h_3\leq t-(h_1+h_4)$ and 
\item $\max\{h_1-h_3-h_4+2e+2-\ell,h_1+h_3+h_4-\ell+1\}\leq h_2\leq \ell-1-(h_1+h_4-h_3)$,
\end{itemize}  then $\mathcal{L}\leq2$ for any $e$-error-correcting code $C$. Next, we modify the presentation we got for $N_3$ into the formulation above.

Let us denote by $i'_2=e+1-i_2$ and by $i'_3=e+1-i_3$. Observe that $\binom{e+1}{i'_j}=\binom{e+1}{i_j}$ for $j\in\{2,3\}$. Notice that we have $i_4\geq0\geq i_1-\ell$  and $\binom{e+1}{i_4}=0$ when $i_4<0$. Hence, we may just replace this lower bound by $i_1-\ell$. Moreover, we have $0\leq i_3\leq \ell-1-i_1$. Hence, $e+1\geq i'_3\geq i_1+e+2-\ell$. Notice that $t-(i_1+i_4)\geq e+1$ since $i_1+i_4\leq \ell-1$ by Inequality (\ref{Ai'_4}), and $\binom{e+1}{i'_3}=0$ when $i'_3>e+1$.  

For $i_2$ we have $i_1+i_3+i_4-(\ell-1)\leq i_2\leq \ell-1/2-i_1-|i_4+1/2-i_3|$ and hence, $\ell-1-(i_1+i_4-i'_3)\geq i'_2\geq -\ell+e+3/2+i_1+|i_4+1/2-i_3|=e+3/2-\ell+i_1+|i_4+i'_3-e-1/2|=\max\{i_1-i'_3-i_4+2e+2-\ell,i_1+i'_3+i_4-\ell+1\}$. Hence, we get the claim.
\end{proof}

\end{document}